\documentclass[11pt]{article}

\usepackage{fullpage}
\usepackage{amssymb,amsmath}
\usepackage{amsthm}
\usepackage{mathdots}
\usepackage{mathpazo}
\usepackage{cases}
\usepackage{graphicx,psfrag,paralist,enumerate}
\usepackage{microtype}
\usepackage{xspace}
\usepackage{marvosym}
\usepackage{tabularx}
\usepackage{csquotes}
\usepackage{color,soul}
\usepackage[vlined,linesnumbered,ruled]{algorithm2e}

\usepackage{thmtools}
\usepackage{thm-restate}

\usepackage{cleveref}

\usepackage{tikz,ifthen,calc}
\usetikzlibrary{positioning}
\usetikzlibrary{shapes}
\usetikzlibrary{shapes.symbols,patterns}
\usetikzlibrary{calc,through,backgrounds}
\usetikzlibrary{decorations.pathreplacing}
\usetikzlibrary{shapes.geometric, arrows}

\parskip=\smallskipamount 

\newtheorem{theorem}{Theorem}

\newtheorem{lemma}[theorem]{Lemma}
\newtheorem{observation}[theorem]{Observation}
\newtheorem{corollary}[theorem]{Corollary}
\newtheorem{claim}{Claim}
\newtheorem{fact}{Fact}

\newcommand{\set}[1]{\left\{ #1 \right\}}
\newcommand{\ang}[1]{\left\langle #1 \right\rangle}

\newcommand{\edgeD}{\Delta_{E}}
\newcommand{\omegaS}{\Delta_{\omega}}
\newcommand{\edgeDbar}{\bar{\Delta}_{E}}
\newcommand{\omegaSbar}{\bar{\Delta}_{\omega}}
\newcommand{\multipl}{\omega}

\newcommand{\cD}{\mathcal{D}}

\newcommand{\ignore}[1]{}
\newcommand{\remove}[1]{}

\def\inline#1:{\par\vskip 3pt\noindent{\bf #1:}\hskip 10pt}
\def\midinline#1:{\par\noindent{\bf #1:}\hskip 10pt}
\def\dnsinline#1:{\par\vskip -7pt\noindent{\bf #1:}\hskip 10pt}

\def\blackslug{\hbox{\hskip 1pt \vrule width 4pt height 8pt
		depth 1.5pt \hskip 1pt}}
\def\QED{\quad\blackslug\lower 8.5pt\null\par}

\long\def\commentstart #1\commentend{}

\def\pos{\mbox{\sf pos}}

\def\BBB{{\sc BBB}}

\long\def\hidestart #1\hideend{}

\title{
\textbf{Approximate Realizations for} \\
\textbf{Outerplanaric Degree Sequences}%
\thanks{This work was supported by US-Israel BSF grant no. 2022205.}
}
\author{
Amotz Bar-Noy%
\thanks{The Graduate Center of the City University of New York, NY 10016, USA.
}
\\
{\tt\small amotz@sci.brooklyn.cuny.edu}
\and
Toni B\"ohnlein%
\thanks{Huawei, Zurich, Switzerland.}
\\
{\tt\small toniboehnlein@web.del}
\and
David Peleg%
\thanks{Department of Computer Science and Applied Mathematics,
The Weizmann Institute of Science, Rehovot 76100, Israel.}
\\
{\tt\small david.peleg@weizmann.ac.il}
\and
Yingli Ran\footnotemark[4]
\\
{\tt\small yingli.ran@weizmann.ac.il}
\and
Dror Rawitz
\thanks{Faculty of Engineering, Bar Ilan University, Ramt-Gan 52900, Israel}
\\
{\tt\small dror.rawitz@biu.ac.il}
}

\begin{document}

\maketitle

\begin{abstract}
We study the question of whether a sequence
$d = (d_1,d_2, \ldots, d_n)$ of positive integers
is the degree sequence of some outerplanar (a.k.a. 1-page book embeddable)
graph $G$.
If so, $G$ is an outerplanar realization of $d$ and $d$ is an outerplanaric
sequence.
The case where $\sum d \leq 2n - 2$ is easy, as $d$ has a realization
by a forest (which is trivially an outerplanar graph).
In this paper, we consider the family $\cD$ of all sequences $d$
of even sum $2n\leq \sum d \le 4n-6-2\multipl_1$, where $\multipl_x$ is
the number of $x$'s in $d$.
(The second inequality is a necessary condition for a sequence $d$
with $\sum d\geq 2n$ to be outerplanaric.)
We partition $\cD$ into two disjoint subfamilies, $\cD=\cD_{NOP}\cup\cD_{2PBE}$,
such that every sequence in $\cD_{NOP}$ is provably non-outerplanaric,
and every sequence in $\cD_{2PBE}$ is given a realizing graph $G$ enjoying
a 2-page book embedding
(and moreover, one of the pages is also bipartite).
%
\end{abstract}

\newpage
\setcounter{page}{1}

\section{Introduction}

\paragraph*{Background.}
%
Our study is concerned with degree sequences of graphs.
The degree of a vertex in a graph $G$ is the number of its incident edges,
and the sequence of vertex degrees of $G$ is denoted by $\deg(G)$.
%
Studied for more than 60 years, the \textsc{Degree Realization} problem
requires to decide, given a sequence $d$ of positive integers,
whether $d$ has a \emph{realizing graph}, namely, a graph $G$
such that $\deg(G)=d$,
and to find such a graph if exists.
A sequence $d$ that has a realizing graph
is called \emph{graphic}.
%
One of the first characterizations of all graphic degree sequences
was given by Erd\"os and Gallai~\cite{EG60}.
About the same time, Havel and Hakimi~\cite{hakimi62,havel55} described
an algorithm that, given a sequence $d$, generates a realization
if $d$ is graphic or verifies that $d$ is not graphic.
Other studies focused on the uniqueness of the realization.
It is known that a sequence can have several non-isomorphic realizations.
As an example, consider the sequence\footnote{We use the shorthand
$a^k$ to denote a subsequence of $k$ consecutive $a$'s.} $(2^6)$, which
can be realized by a 6-vertex cycle or by two 3-vertex cycles.
This issue is not of our concern in the current paper.

The realization problem has also been investigated for
various well known classes of graphs.
In particular, realization by planar graphs was considered
in several papers, but currently a complete solution seems out of reach.
The survey by Rao~\cite{Rao81} listed a number of partial results and
studies on related questions.
In this paper, we address the question
whether a given sequence has an \emph{outerplanar} realization.
Outerplanar graphs are
planar graphs that can be embedded in the plane so that edges do not intersect
each other and additionally,
each vertex lies on the outer face
(i.e., no vertex is surrounded by edges)
\cite{syslo1979characterizations}.
%

A sequence $d\in\cD$ is \emph{outerplanaric}
if it has a realizing outerplanar graph.
Similarly, a sequence $d\in\cD$ is \emph{maximal outerplanaric (MOP)}
if it has a realizing \emph{maximal} outerplanar graph
(namely, an outerplanar graph with the property that adding any
edge to it violates outerplanarity).

Throughout, we denote by $d = (d_1, d_2, \ldots, d_n)$
a nonincreasing sequence of $n$ positive integers and denote
its \emph{volume} by  $\sum d = \sum_{i=1}^n d_i$.
We use the shorthand $a^k$ to denote a subsequence of $k$
consecutive integers $a$, and denote the multiplicity of degree $x$
in a sequence $d$ by $\multipl_x$.
By the well-known Handshaking Lemma, $\sum d$ is even for every
graphic sequence $d$.
If $\sum d \leq 2n - 2$ and $\sum d$ is even, then $d$ has
a forest realization with $(2n - \sum_i d_i)/2$  components~\cite{GJT07}.
Therefore, we focus on $d$ with $\sum d\ge 2n$.
We show later (Lemma \ref{lem:outerplanar-sum-bound-omega1})
that if $d$ is outerplanaric but not forestic (i.e., $\sum d > 2n - 2$), then
also
$\sum d \leq 4n - 6 - 2\multipl_1$,
and moreover, $\sum d=4n-6$ if and only if $d$ is maximal outerplanaric.
Consequently, we focus on the family of nonincreasing sequences of $n$
positive integers
\[\cD ~=~ \left\{d=(d_1,\ldots,d_n)
~\mid~ \sum d ~\mbox{is even and}~ 2n\leq \sum d\le 4n-6-2\multipl_1 \right\}.\]
Let $\cD_{\le 4}\subseteq \cD$ be the subclass of sequences
whose largest degree satisfies $d_1\le 4$.

While we are unable to fully characterize the sequences of the family
$\cD$ with respect to outerplanarity, we obtain a classification result that can be
thought of as an ``approximate'' solution to the outerplanarity problem.
To do that, we employ the concept of \emph{book embeddings} which is used
in graph theory to measure how far a graph is \emph{away} from being outerplanar.
Given a graph $G=(V,E)$, let $E = E_1 \cup \ldots \cup E_p$ be
a partition of its edges such that each
subgraph $G_i = (V,E_i)$ is outerplanar.
%
For a book embedding of $G$, think of a book in which the pages (half-planes)
are filled by outerplanar embeddings of the $G_i$'s such that the vertices
are embedded on the spine of the book and in the same location on each page.
This constraint is equivalent to requiring that the vertices appear in the same
order along the cyclic order
(defined later as the vertex ordering on the outer face)
of each of the outerplanar embeddings of the $G_i$'s.
The \emph{book thickness} or \emph{pagenumber}~\cite{bernhart1979book} is
the minimal number of pages for which a graph has a valid book embedding.
Note that a graph is outerplanar if and only if it has pagenumber $1$,
and it is known that the pagenumber of planar graphs
is at most $4$~\cite{Yannakakis-stoc86}.

\paragraph*{Our contribution.}
%
Our main result is
a classification of $\cD$ into two sub-families, $\cD=\cD_{NOP}\cup\cD_{2PBE}$,
such that
\begin{compactitem}
\item every sequence in $\cD_{NOP}$ is provably non-outerplanaric, and
\item every sequence in $\cD_{2PBE}$ is given (in polynomial time)
a realizing graph $G=(V,E)$,~ $E=E_1\cup E_2$,
enjoying a 2-page book embedding with pages
$G_1$ and $G_2$
(hereafter referred to as a \emph{2PBE realization}). Moreover,
$G_2$
is bipartite
(hereafter, an \emph{OP+bi realization}),
and in some cases,
$G_2$ consists of at most one or two edges
(hereafter, an \emph{OP+1} or \emph{OP+2} realization).
\end{compactitem}
%
Note that every \emph{outerplanaric} sequence in $D$ falls in $\cD_{2PBE}$.
Thus, the above result can be interpreted as yielding a polynomial time
algorithm providing an ``approximation'' for outerplanaric sequences,
in the sense that if the given sequence $d$ is outerplanaric
(i.e., has a 1-page realization), then the algorithm will produce
a 2-page realization for it.
We remark that for many of the outerplanaric sequences in $\cD$,
we are in fact able to find an outerplanar realization,
but as this refined classification is incomplete, it is omitted here.

Figure~\ref{fig:decision-tree} gives an overview of the algorithm.


\tikzstyle{startstop} = [rectangle, rounded corners,
minimum width=3cm,
minimum height=1cm,
text centered,
draw=black,
fill=red!30]

\tikzstyle{io} = [trapezium,
trapezium stretches=true, 
trapezium left angle=70,
trapezium right angle=110,
minimum width=3cm,
text centered,
draw=black, fill=blue!30]

\tikzstyle{process} = [rectangle,
minimum width=2.5cm,
text centered,
text width=2.5cm,
draw=black,
fill=orange!30]

\tikzstyle{decision} = [diamond,
minimum width=4cm,
aspect=4,
text centered,
draw=black,
fill=green!30]

\tikzstyle{arrow} = [thick,->,>=stealth]

\begin{figure}[t]
\centering
\begin{scriptsize}
\begin{tikzpicture}[node distance=2cm]

\node (io1) [io] {$d \in \mathcal{D}$, $d$ is graphic};
\node (dec1) [decision, below of=io1, yshift=0.6cm] {$d_{n-1} \leq 2$ and $d_{n-2} \leq 3$};
\node (pro1) [process, left of=dec1, xshift=-2.5cm] {Not outerplanaric \\ (\Cref{lem:op_nec_deg23})};
\node (dec2) [decision, below of=dec1, yshift=0.4cm] {$d \in \mathcal{D}_{\leq 4}$?};
\node (pro2) [process, right of=dec2, xshift=2.5cm] {Build 2PBE (\Cref{thm:d1leq4})};
\node (dec3) [decision, below of=dec2, yshift=0.6cm] {$d_2 \leq 3$?};
\node (pro3) [process, right of=dec3, xshift=2.5cm] {Build 2PBE (\Cref{small d2})};
\node (dec4) [decision, below of=dec3, yshift=0.6cm] {$2\Delta_E \geq \Delta_\omega$};
\node (pro4) [process, right of=dec4, xshift=2.5cm] {Build 2PBE (\Cref{outerplanar with degree 1})};
\node (pro5) [process, left of=dec4, xshift=-2.5cm] {Not outerplanaric \\ (\Cref{lem: between edge and multiplicity})};

\draw [arrow] (io1) -- (dec1);
\draw [arrow] (dec1) -- node[anchor=south] {no} (pro1);
\draw [arrow] (dec1) -- node[anchor=east] {yes} (dec2);
\draw [arrow] (dec2) -- node[anchor=south] {yes} (pro2);
\draw [arrow] (dec2) -- node[anchor=east] {no} (dec3);
\draw [arrow] (dec3) -- node[anchor=south] {yes} (pro3);
\draw [arrow] (dec3) -- node[anchor=east] {no} (dec4);
\draw [arrow] (dec4) -- node[anchor=south] {yes} (pro4);
\draw [arrow] (dec4) -- node[anchor=south] {no} (pro5);
\end{tikzpicture}
\end{scriptsize}
\caption{\small\sf
A flowchart of the resulting algorithm.
$\edgeD= (4n - 6 - \sum d) / 2$ and
$\omegaS=3\multipl_1+2 \multipl_2 + \multipl_3 - n$.
}
\label{fig:decision-tree}
\end{figure}
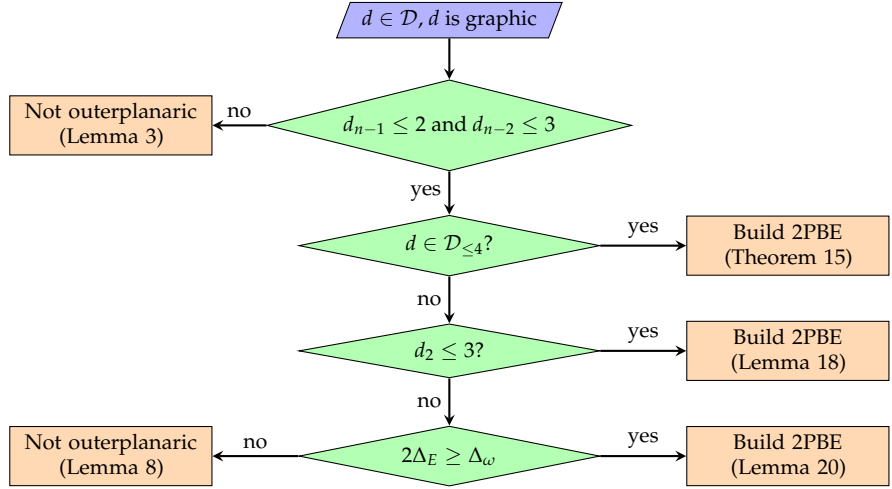


%
%
%

\paragraph*{Additional related work.}
%
%
For the \textsc{Outerplanar Degree Realization} problem, a full characterization
of \emph{forcibly} outerplanar graphic sequences
(namely, sequences each of whose realizations is outerplanar)
was given in~\cite{choudum1981characterization}.
%
%
A characterization of the degree sequence of maximal outerplanar graphs
having exactly two 2-degree nodes was provided in~\cite{rafid2022generating}.
This was also mentioned in~\cite{bose2008characterization}.
A characterization of the degree sequences of maximal outerplanar graphs with
at most four vertices of degree 2 was given in~\cite{li2018degree}.

Sufficient conditions for the realization of $2$-trees were given in~\cite{LMNW06}.
A full characterization of the degree sequences of 2-trees was presented
in~\cite{bose2008characterization}.
A graph $G$ is a 2-tree if $G = K_3$ or $G$ has a vertex $v$ with degree 2,
whose neighbors are adjacent and $G[V \setminus \set{v}]$ is a 2-tree.
This definition implies, in particular, that 2-trees have $\sum d=4n-6$.
By Theorem 1 of~\cite{bose2008characterization}, if a sequence $d$ is
maximal outerplanaric, i.e., it is outerplanaric and satisfies $\sum d=4n-6$,
then $d$ can be realized by a 2-tree.
In~\cite{rengarajan1995stack}
it was shown that every 2-tree has a 2-page book embedding.
Note that the class of 2-trees is not hereditary, i.e.,
subgraphs of 2-trees are not necessarily 2-trees.
Therefore, the result of~\cite{bose2008characterization}
does not apply to non-maximal degree sequences.

For the \textsc{Planar Degree Realization} problem, regular planar graphic
sequences were classified in~\cite{hawkins1966certain}, and
planar bipartite biregular degree sequences were studied in~\cite{AdamsN19}.
In~\cite{SH77},  Schmeichel and Hakimi determined which
graphic sequences with $d_1-d_n=1$ are planaric, and
presented similar results for $d_1-d_n=2$,
with a small number of unsolved cases.
Some of the sequences left unsolved in~\cite{SH77} were later resolved
in~\cite{fanelli1980conjecture-maximal,fanelli1981unresolved-nonmaximal}.
Some additional studies on special cases of the planaric degree
realization problem were listed in~\cite{Rao81}.

\section{Preliminaries}
\label{sec:2}

\paragraph{Definitions and notation.}
%
Given two $n$-integer sequences $d$ and $d'$, denote their componentwise
addition and difference, respectively, by
$d\oplus d' = (d_1+d'_1,\ldots,d_n+d'_n)$ and
$d\ominus d' = (d_1-d'_1,\ldots,d_n-d'_n)$.
For a nonincreasing sequence $d$ of $n$ nonnegative integers,
let $\pos(d)$ denote the prefix consisting of the positive integers of $d$
(discarding the zeros).
Given two sequences $d'$ and $d''$ of length $n'$ and $n''$ respectively, let
$d'\circ d'' = (d'_1,\ldots,d'_{n'},d''_1,\ldots,d''_{n''})$.


Consider a simple (undirected) graph $G = (V,E)$.
A (directed) \emph{circuit} in $G$ is an \emph{ordered set} of vertices
$C = \set{v_0,v_1,\ldots,v_{k-1}}$ (with possible repetitions) such that
$(v_i,v_{(i+1) \bmod k})\in E$ for every $i=0,\ldots,k-1$.
The directed edge set of the circuit $C$ is
$E(C) = \{\langle v_i,v_{(i+1) \bmod k}\rangle \mid i=0,\ldots,k-1\}$
and it is required
that no directed edge may appear on it more than once.
If the vertices of $C$ are distinct then it is called a \emph{cycle}.
A \emph{triangle} is a cycle consisting of three edges.

For a graph $G = (V,E)$, one may define an accompanying embedding in the plane,
namely, a mapping of every vertex to a point in two-dimensional space, with
a curve connecting the two endpoints $u,w$
of every edge $(u,w)\in E$.
%
%
An \emph{outerplanar} graph $G$ is a graph that has a planar embedding
in which all the vertices occur on the outer face.
(This embedding can be transformed into a 1-page book embedding of $G$,
with all vertices placed on the \emph{spine} of the book.)
We refer to the order in which the vertices occur on the outer face as the
\emph{cyclic-order}
of $G$, and denote it by $C_{order}(G)$.
A \emph{maximal outerplanar} (MOP) graph $G$ is an outerplanar graph
such that adding any new edge to it results in a non-outerplanar graph.
It is
folklore
%
%
%
that for a
2-connected outerplanar graph,
the cyclic order corresponds to a Hamiltonian cycle in $G$.

Hereafter, whenever we discuss a given outerplanar (resp., 2PBE) graph,
we implicitly assume that some outerplanar (resp., 2PBE) embedding for it
is given as well, and when constructing an outerplanar (resp., 2PBE)
graph, we also give an outerplanar (resp., 2PBE) embedding for it.

Given an outer-planar embedding, an \emph{external} edge
is an edge residing on the outer face.
That is,  $e=(u,w)$ is external if $u$ and $w$ are neighbors
on a circuit which is part of the outer face.
All other edges of $G$ are \emph{internal}.
An \emph{outer} (respectively, \emph{inner}) \emph{cycle} in $G$ is a cycle
all of whose edges are external (resp., internal).
Note that the outer face may be comprised of several outer-cycles
(as can be seen in Figure~\ref{f:realization}).


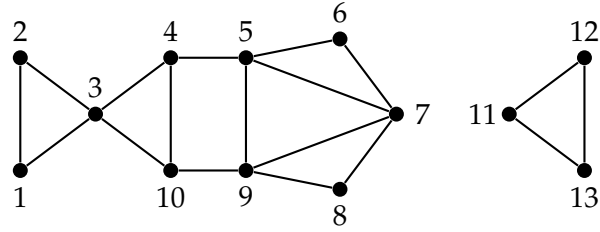
\begin{figure}
\centering
\begin{tikzpicture}[scale=0.5]
\begin{scope}[
every node/.style={fill=black,circle,inner sep=2pt},
every edge/.style={draw=black,thick}]
\node(v1) at (0,0) {};
\node(v2) at (0,3) {};
\node(v3) at (2,1.5) {};
\node(v4) at (4,3) {};
\node(v5) at (6,3) {};
\node(v6) at (8.5,3.5) {};
\node(v7) at (10,1.5) {};
\node(v8) at (8.5,-0.5) {};
\node(v9) at (6,0) {};
\node(v10) at (4,0) {};
\node(v11) at (13,1.5) {};
\node(v12) at (15,3) {};
\node(v13) at (15,0) {};

\path (v1) [-] edge (v2);
\path (v1) [-] edge (v3);
\path (v2) [-] edge (v3);
\path (v3) [-] edge (v4);
\path (v3) [-] edge (v10);
\path (v4) [-] edge (v10);
\path (v4) [-] edge (v5);
\path (v5) [-] edge (v9);
\path (v5) [-] edge (v6);
\path (v5) [-] edge (v7);
\path (v6) [-] edge (v7);
\path (v7) [-] edge (v8);
\path (v7) [-] edge (v9);
\path (v8) [-] edge (v9);
\path (v9) [-] edge (v10);
\path (v11) [-] edge (v12);
\path (v11) [-] edge (v13);
\path (v12) [-] edge (v13);
\end{scope}

\node () at (0,-0.7) {$1$};
\node () at (0,3.7) {$2$};
\node () at (2,2.2) {$3$};
\node () at (4,3.7) {$4$};
\node () at (6,3.7) {$5$};
\node () at (8.5,4.2) {$6$};
\node () at (10.7,1.5) {$7$};
\node () at (8.5,-1.2) {$8$};
\node () at (6,-0.7) {$9$};
\node () at (4,-0.7) {$10$};
\node () at (12.3,1.5) {$11$};
\node () at (15,3.7) {$12$};
\node () at (15,-0.7) {$13$};
\end{tikzpicture}
\caption{\small\sf An outerplanar graph without degree-1 vertices.
The triangle $(5,7,9)$ is the only internal cycle,
while $(1,2,3)$, $(3,4,5,6,7,8,9,10)$, and $(11,12,13)$ are outer-cycles.
A possible cyclic-order is $(1,2,3,4,5,6,7,8,9,10,11,12,13)$.
}
\label{f:realization}
\end{figure}

\paragraph{Degree sequences of trees and forests.}
Consider a nonincreasing sequence $d = (d_1, \ldots, d_n)$ of positive integers
such that $\sum d$ is an even number.
It is known that if $\sum d \leq 2n - 2$ then $d$ is graphic, and
moreover, $d$ can be realized by an acyclic graph (forest).
In this case, $d$ is called a \emph{forestic} sequence.
If $\sum d = 2n - 2$, then $d$ can be realized only by trees,
and $d$ is called a \emph{treeic} sequence.
A vertex of degree one is called a \emph{leaf}.

We make use of a special type of realizations of treeic
sequences by \emph{caterpillar graphs}.
In a caterpillar graph $G = (V,E)$, all the non-leaves vertices
are arranged on a path,
called the \emph{spine} of the caterpillar. Formally, the spine of
a caterpillar graph is an ordered sequence $S = (x_1, \ldots, x_s) \subseteq V$,
$s\ge 1$, in which $(x_i,x_{i+1}) \in E$ for $i = 1, \ldots, s-1$
(see Figure~\ref{fig:caterpillar-example} for an example).

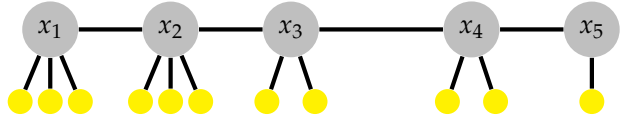
\begin{figure}[t]
\centering
\begin{small}
	\begin{tikzpicture}[scale=0.8]
			
		\def\pathlen{5}
		\def\pathspacing{2.3}
		\pgfmathtruncatemacro{\plenplusone}{\pathlen + 1}
		
		\foreach \x in {1,...,\pathlen} {
		\pgfmathtruncatemacro{\xcoord}{\x * \pathspacing}
		\node [minimum size= 3mm,fill=lightgray,circle] (v\x) at (\xcoord,0)  {$x_{\x}$};
		}
				
	\pgfmathtruncatemacro{\xcoord}{1 * \pathspacing}				\node [fill=yellow,circle] (leaf1) at ( \xcoord - .5, - 1.2)  {};
	\path [draw=black,ultra thick] (leaf1) [-] edge (v1);
					
	\node [fill=yellow,circle] (leaf2) at ( \xcoord, - 1.2)  {};
	\path [draw=black,ultra thick] (leaf2) [-] edge (v1);
					
	\node [fill=yellow,circle] (leaf3) at ( \xcoord + .5, - 1.2)  {};
	\path [draw=black,ultra thick] (leaf3) [-] edge (v1);
							
	\pgfmathtruncatemacro{\xcoord}{2 * \pathspacing}					
	\node [fill=yellow,circle] (leaf1) at ( \xcoord - .5, - 1.2)  {};
	\path [draw=black,ultra thick] (leaf1) [-] edge (v2);
					
	\node [fill=yellow,circle] (leaf2) at ( \xcoord, - 1.2)  {};
	\path [draw=black,ultra thick] (leaf2) [-] edge (v2);
					
	\node [fill=yellow,circle] (leaf3) at ( \xcoord + .5, - 1.2)  {};
	\path [draw=black,ultra thick] (leaf3) [-] edge (v2);		
							
	\pgfmathtruncatemacro{\xcoord}{3 * \pathspacing}

	\node [fill=yellow,circle] (leaf1) at ( \xcoord - .4, - 1.2)  {};
	\path [draw=black,ultra thick] (leaf1) [-] edge (v3);
					
	\node [fill=yellow,circle] (leaf3) at ( \xcoord + .4, - 1.2)  {};
	\path [draw=black,ultra thick] (leaf3) [-] edge (v3);		
				
	\pgfmathtruncatemacro{\xcoord}{4 * \pathspacing}				
	\node [fill=yellow,circle] (leaf1) at ( \xcoord - .4, - 1.2)  {};
	\path [draw=black,ultra thick] (leaf1) [-] edge (v4);
							
	\node [fill=yellow,circle] (leaf3) at ( \xcoord + .4, - 1.2)  {};
	\path [draw=black,ultra thick] (leaf3) [-] edge (v4);		
					
	\pgfmathtruncatemacro{\xcoord}{5 * \pathspacing}					
	\node [fill=yellow,circle] (leaf1) at ( \xcoord, - 1.2)  {};
	\path [draw=black,ultra thick] (leaf1) [-] edge (v5);
					
	\pgfmathtruncatemacro{\plenminusone}{\pathlen - 1}
	\foreach \x in {1,...,\plenminusone}{
						
		\pgfmathtruncatemacro{\xpp}{\x+1}
		\path [draw=black,ultra thick] (v\x) edge (v\xpp);
						
	}
\end{tikzpicture}
\end{small}
\caption{\small\sf A caterpillar graph with degree sequence $(5,4^3,2,1^{11})$.
Spine vertices are gray, leaves are yellow.
}
\label{fig:caterpillar-example}
\end{figure}

The other ingredient in the construction of a realization to a forestic
sequence is a collection of disjoint edges, hereafter referred to
somewhat informally as a \emph{matching}.


\begin{observation}
\label{obs:caterpillar}
A forestic sequence $d = (d_1,d_2, \ldots, d_n)$ of positive integers can be
realized by a union of a caterpillar graph and a matching.
\end{observation}

\begin{proof}
First consider a treeic sequence $d$ (such that $\sum d = 2n - 2$).
Assume there are $n-s$ vertices of degree 1 in $d$.
Denote by $d^*$ the prefix of $d$ that contains all the degrees $d_i>1$.
It follows that $\sum d^* + (n - s)  = \sum d = 2n - 2$, or equivalently
$\sum d^* - 2s + 2 = n - s$.
To get the caterpillar realization of $d$, first arrange $s$ vertices
corresponding to the degrees of $d^*$ in a path.
The path edges contribute $2s - 2$ to the volume $\sum d^*$.
The missing $\sum d^* - 2s + 2 = n - s$ degrees are satisfied by attaching the
$n-s$ leaves to the path, which now forms the spine of the caterpillar
realization. Note that the order of the vertices on the spine can be arbitrary.

Now assume that $d$ is forestic but not treeic, i.e., $\sum d < 2n - 2$.
In this case, remove pairs of
1 degrees from $d$ until the volume of the reaming sequence $d'$ with
$n'$ vertices is $2(n'-2)$.
This must happen because each pair removal reduces the volume by $2$
while $2n-2$ decreases by $4$. Let $G'$ be the caterpilar realization of $d'$
as implied by the first part of the proof. To get the realization of $d$,
add $(n-n')/2$ edges to $G'$ to satisfy the $n-n'$ removed 1 degrees
by a matching.
\end{proof}

\paragraph{Some known results.}
%
We next summarize known conditions for sequence outerplanarity.

\begin{lemma}[\bf \cite{West01}]
\label{lem:euler-sequence-outerplanar}
If $d = (d_1 , \ldots, d_n)$ is an outerplanaric degree sequence where $n \geq 2$, then
$\sum d \leq 4 n - 6$,
with equality if and only if $d$ is maximal outerplanaric.
\end{lemma}


\begin{lemma}[\bf \cite{syslo1979characterizations}]
\label{lem:op_nec_deg23}
If $d = (d_1 , \ldots, d_n)$ is
outerplanaric, then
\begin{inparaenum}[(i)]
\item $d_{n-1} \leq 2$, and
\item $d_{n-2} \leq 3$.
\end{inparaenum}
Moreover, $d_{n-1}=d_n=2$ in a maximal outerplanaric sequence $d$.
\end{lemma}

Our next observation is that Lemma~\ref{lem:euler-sequence-outerplanar}
can be improved if $\multipl_1$, the multiplicity of degree $1$ in $d$
is taken into consideration. 
Specifically, for the case $\sum d > 2n-2$, we have the following.

\begin{lemma}
\label{lem:outerplanar-sum-bound-omega1}
If
$d = (d_1, \ldots, d_n)$
is outerplanaric and $\sum d > 2n-2$, then $\sum d \leq 4n - 6 - 2\multipl_1$.
\end{lemma}

\begin{proof}
Let $d$ be as in the lemma and let $G = (V,E)$ be an outerplanar graph
that realizes $d$.	
Construct a graph $G'$ by deleting from $G$ all the $\multipl_1$ vertices
of degree $1$ together with their incident edges, and subsequently deleting the vertices whose degree was reduced to zero.
The graph $G'$ may contain vertices whose degree became 1.
Observe that $G'$ is outerplanar, and let $d' = \deg(G')$.
Note that $n' \leq n - \multipl_1$, where $n'$ the number of vertices in $G'$.
Lemma~\ref{lem:euler-sequence-outerplanar} yields the claim, since ~~
\[
\sum d
= \sum d' + 2 \multipl_1 \leq 4(n-\multipl_1) - 6 + 2 \multipl_1
= 4n - 6  - 2 \multipl_1
~,
\]
where the first equality comes from the fact that $\multipl_1$ vertices
of degree $1$ together with their incident edges are deleted,
and the inequality is due to Lemma \ref{lem:euler-sequence-outerplanar}.
\end{proof}

\noindent
We remark that Lemma~\ref{lem:outerplanar-sum-bound-omega1} is false
without the assumption that $\sum d > 2n - 2$.
To see this, consider the treeic sequence $d = (a,1^a)$ which is (uniquely)
realized by a star graph and therefore is outerplanaric.
The process in the proof of Lemma~\ref{lem:outerplanar-sum-bound-omega1}
yields $G'$ consisting of one isolated vertex.
Note that Lemma~\ref{lem:euler-sequence-outerplanar} cannot be applied
in this case. Indeed, we have
\[\sum d = 2a > 2a-2 = 4(a+1) - 6 - 2a = 4n - 6 - 2\multipl_1~.\]

\smallskip
We now recall some facts concerning
maximal outerplanaric sequences (or MOP's).

\begin{fact}[\bf \cite{proskurowski1979minimum}]
\label{lem: MOP construction}
A graph is a MOP
if and only if
it can be constructed from a triangle by
a finite number of applications of
repeatedly  applying the following operation:
to the graph already constructed, add a new vertex in the exterior face
and join it to two adjacent vertices on
that face.
\end{fact}

We rely on the construction method of Lemma \ref{lem: MOP construction}
to prove the following simple but useful facts.
Parts (1) and (2) are known (cf. \cite{West01,campos2013dominating}.
We could not find a reference for Part (3), so we give a proof for completeness.



\begin{lemma}
\label{obs: MOP-Hamiltonian-cycle}
Consider a maximal outerplanar graph $G$ with $n\ge 4$ vertices.
\begin{compactenum}[(1)]
\item Every two vertices of degree 2 in $G$ are nonadjacent.
\item there is a outerplanar embedding of $G$ such that the boundary of the
  outer face is a Hamiltonian cycle $C$ and each inner face is a triangle.
\item Removing the edges of $C$ from $G$ and discarding the resulting degree-0
vertices yields a connected outerplanar graph.
\end{compactenum}
\end{lemma}


\begin{proof}
The lemma is proved by induction on the number of vertices $n$.
The claim clearly holds for the base case of $n=4$, where the 4-cycle
augmented by a single chord is the only MOP.
Now suppose the claim holds for all MOP graphs with at most $n$ vertices,
and consider a graph $G$ with $n+1$ vertices.

By \Cref{lem: MOP construction} $G$ was obtained by adding a vertex $v$
and edges $(v,u)$ and $(v,w)$ to some MOP graph $G'$ in which
$e=(u,w)$ is an external edge.
By the inductive hypothesis, the claim holds for $G'$, so there exists
an outer-cycle $C'$ of $G'$ which is Hamiltonian, and removing $C'$ yields
a connected graph $c(G')$.
Note that $u$ and $w$ both appear in $G'$, so by \Cref{lem:op_nec_deg23}
their degrees in $G'$ are 2 or higher, and moreover, at least one of
$\deg_{G'}(u)$ and $\deg_{G'}(w)$ is 3 or higher, since $G'$ satisfies
Property (1). This implies that at least one of $u$ and $w$ belongs to $c(G')$.
Adding $v$ to $G$ and connecting it to $u$ and $w$ increases their degrees
in $G$ to 3 or higher, so $v$ has no neighbor of degree 2 in $G$.
Hence, property (1) holds in $G$ by the inductive hypothesis.
Also, inserting the vertex $v$ between $u$ and $w$ in $C'$
yields a Hamiltonian cycle $C$ for $G$, implying Property (2).
Finally, removing $C$ from $G$ yields the graph $c(G)=c(G')\cup\{e\}$,
which is connected since $c(G')$ is connected and it contains $u$ or $w$.
\end{proof}

%
Jao and West~\cite{Jao2010West} show that the number of internal triangles
of a MOP is related to $\multipl_2$.

\begin{lemma}{\bf \cite{Jao2010West}}
\label{lem:MOP-deg2-triangles}
Let $G$ be a maximal outerplanar graph on $n$ vertices, let $d = \deg(G)$,
and let $t$ be the number of internal triangles.
If $n \geq 4$, then $t = \multipl_{2} - 2$.	
\end{lemma}

\paragraph*{Bipartite bracket balancing (BBB) algorithm.}
we make use of a simple variant of the \textsc{Bracket Balancing} problem,
in which the matched brackets must fit in given positions.
Formally, the problem, referred to as
\textsc{Bipartite Bracket Balancing (BBB)}, is defined as follows.
We are given two ordered sets $P= \ang{p_1,\ldots,p_k}$ and $Q = \ang{ q_1,\ldots,q_k }$, merged into a $2k$-element sequence $R = \ang{ r_1,\ldots,r_{2k} }$ that contains all the elements from $P$ and $Q$ while maintaining their original order.
The goal is to associate with every $r_i$, $1\le i\le 2k$, a bracket $f(r_i) \in \set{[,]}$ with the following two properties:
\begin{compactitem}
\item The sequence $f(R) = \langle f(r_1),\ldots,f(r_{2k})\rangle$ is balanced
  such that the number of opening brackets ($``[''$) in every prefix of $R$ is no less than the number of closing brackets ($``]''$) in this prefix. Observe that in balanced sequences it is possible to match the opening brackets with the closing brackets
so that the intervals associated with the matched brackets are either nested or disjoint.
\item if $f(r_i)=[$ and $f(r_j)=]$ were matched, for $i<j$, then $r_i\in P$ and
$r_j \in Q$ or vice versa.
\end{compactitem}


This problem can be solved by a straightforward variant of the standard
stack-based linear-time bracket balancing algorithm, named Algorithm \BBB,
described here for completeness.
The algorithm uses a stack $D$ and scans the sequence $R$ from start to end.
At any given time, the stack contains either only elements from $P$
or only elements from $Q$ (i.e., $D\subseteq P$ or $D\subseteq Q$).
Upon reaching the element $r_i$, the algorithm acts based on the stack
contents, as follows:
\begin{itemize}
\item $D=\emptyset$, or $D\subseteq P$ and $r_i\in P$,
or $D\subseteq Q$ and $r_i\in Q$~: ~~~
\\
\hbox{~}\hspace{30pt}
$PUSH(r_i,D)$
\item $D\subseteq P$ and $r_i\in Q$,
or $D\subseteq Q$ and $r_i\in P$~: ~
\\
\hbox{~}\hspace{30pt}
$x \gets POP(D)~~$;~~~~~ $f(x)\gets ~[$~~;~~~~~ $f(r_i)\gets ~]$
\end{itemize}

\noindent
It is easy to verify that the algorithm solves
the \textsc{Bipartite Bracket Balancing} problem.
%
\begin{figure}[h]
\begin{center}
\includegraphics[width=4in]{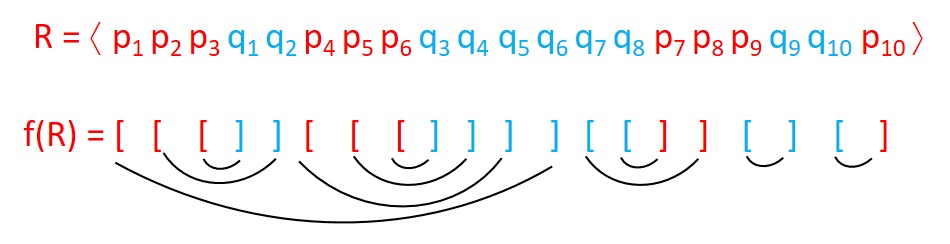}
\end{center}
\caption{\small\sf A merged sequence $R$ and its corresponding bracket assignment.}
\label{fig:bipartite-balanced-brackets}
\end{figure}

Figure \ref{fig:bipartite-balanced-brackets} illustrates an example sequence
$R$ and the corresponding balanced bracket assignment generated by the
algorithm. Notice that the stack switches its contents from time to time.
In particular,
it initially gets only elements of $P$,
then becomes empty after inspecting $r_{12}=q_6$, then contains only elements
of $Q$ for a while, then becomes empty after inspecting $r_{16}=q_8$, and so on.




\section{Necessary and Sufficient Conditions for Outerplanarity}

This section presents some necessary and sufficient conditions for outerplanarity, which are
used to prove the results of the main section concerning 2-page book embedding.
The next lemma provides a necessary condition for maximal outerplanarity.

\begin{lemma}
\label{lem:maxop-2omega2-omega3-general}
If $d = (d_1, \ldots, d_n)$
is a maximal outerplanaric sequence where $d_1 > 3$
%
%
then
if $d_2>3$ or $\multipl_2>2$ then $2 \multipl_2 + \multipl_3 \leq n$
otherwise $2 \multipl_2 + \multipl_3 \leq n+1$.
\end{lemma}

\begin{proof}
Let $d$ be as in the lemma and let $G$ be a maximal outerplanar realization of $d$.
Observe that by \Cref{obs: MOP-Hamiltonian-cycle}(2),
removing the outer-cycle from $G$ and discarding the resulting
0-degree vertices results in an outerplanar graph $G'$ with degree sequence
$d'=\pos(d\ominus (2^n))$ such that $n'=n-\multipl_2$.
It follows that $d'$ is outerplanaric.
We consider the following three cases.
\begin{description}
\item[Case 1:] $\multipl_2 > 2$. $G$ has an internal triangle due to
\Cref{lem:MOP-deg2-triangles}.
This triangle still exists in $G'$. Moreover, $G'$ is connected
by \Cref{obs: MOP-Hamiltonian-cycle}(3).
It follows that $\sum d' > 2n'-2$.
By \Cref{lem:outerplanar-sum-bound-omega1},
\[
2n - 6
~=~    \sum d - 2n
~=~    \sum d'
~\leq~ 4n' - 6 - 2 \multipl_{1}^{'}
~=~    4(n-\multipl_{2}) - 6 - 2 \multipl_{3}.
\]
Simplifying yields that $2\multipl_2 + \multipl_3 \leq n$.

\item[Case 2:] $\multipl_2=2$ and $d_2 > 3$. We have that $\multipl_3 \leq n - 4$
and $2 \multipl_2 + \multipl_3 = 4 + \multipl_3 \leq n$.

\item[Case 3:] $\multipl_2 = 2$ and $d_2\leq 3$.	
Since $d_1 > 3$, we have that $\multipl_3 \leq n - 3$.
It follows that
$2 \multipl_2 + \multipl_3 = 4 + \multipl_3 \leq n+1$.
\end{description}
The lemma follows.
\end{proof}

\noindent
Note that this lemma does not hold for non-maximal outerplanaric
sequences. For example, the outerplanaric sequence $(4,2^4)$
is a counterexample for which $2\multipl_2+\multipl_3=8=n+3$.

We next define two functions that play a key role in our analysis.
For a sequence $d$, define the \emph{edge deficit}
and \emph{low degree surplus} of $d$ to be, respectively,
\begin{align*}
\edgeD(d)   & = (4n - 6 - \sum d) / 2,
&
\omegaS(d) &=  3\multipl_1+2 \multipl_2 + \multipl_3 - n.
\end{align*}

Note that $\edgeD(d) \geq 0$ for any outerplanaric sequence.
If $\edgeD(d) > 0$ and $G$ is a nonmaximal outerplanar realization of $d$, then
$\edgeD(d)$ is the number of edges that need to be added to $G$ to generate
a maximal outerplanar graph.
If $d$ is clear from the context we write simply $\edgeD,\omegaS$ instead of $\edgeD(d),\omegaS(d)$.
The next lemma shows a connection between $\edgeD$ and $\omegaS$.


\begin{lemma}
\label{lem: between edge and multiplicity}
Let $d = (d_1, \ldots, d_n)$
be an outerplanaric sequence with $d_2 > 3$.
Then $2\edgeD \geq \omegaS$.
\end{lemma}
\begin{proof}
Let $d$ be as in the lemma, and let $G$ be an outerplanar realization of $d$.
We can add $\edgeD$ edges to $G$ so as to transform
it into a maximal outerplanar graph $G'$ (without adding vertices).
For $0\leq i\leq \edgeD$,
let $G_i$ be the graph $G$ after adding $i$ edges,
and let $d^i$ be the degree sequence of $G_i$.
By definition, $G_{0}=G$ and $G_{\edgeD}=G'$.
Define $f(G_i) = 2\edgeD(d^i) - \omegaS(d^i)$ for $0\leq i\leq \edgeD$.
We show:
\begin{align}
\label{goal}
f(G_i)\geq f(G_{i+1}) \textrm{~for ~any} ~0\leq i\leq \edgeD - 1.
\end{align}
Let $d' = \deg(G')$.
Then $\sum d' = 4n-6$ and $\multipl'_1=0$.
\Cref{lem:maxop-2omega2-omega3-general} implies that
$2\multipl'_2 + \multipl'_3 \leq n$. Thus,
$\edgeD(d')=0$ and $\omegaS(d')=3\multipl'_1+2 \multipl'_2 + \multipl'_3 - n\leq 0$. Hence,
\begin{align}
\label{base ground}
f(G_{\edgeD})=2\edgeD(d') - \omegaS(d')\geq 0.
\end{align}
%

\begin{table}
\centering
\begin{tabular}{|l|r|}
\hline
$\deg_{G_i}(u),\deg_{G_i}(v) $	&  $\delta$ \\
\hline
$(1,1) \to (2,2)$  & $-2$\\
\hline
$(1,2) \to (2,3)$ & $-2$\\
\hline
$(1,3) \to (2,4+)$ & $-2$\\
\hline
$(1,4+) \to (2,5+)$ & $-1$\\
\hline
$(2,2) \to (3,3)$	&  $-2$\\
\hline
\end{tabular}
~
\hspace{30pt}
~
\begin{tabular}{|l|r|}
\hline
$\deg_{G_i}(u),\deg_{G_i}(v) $	&  $\delta$ \\
\hline
$(2,3)  \to (3,4+)$	&  $-2$\\
\hline
$(2,4+)  \to (3,4+)$	&  $-1$\\
\hline
$(3,3)  \to (4+,4+)$	&  $-2$\\
\hline
$(3,4+)  \to (4+,4+)$	&  $-1$\\
\hline
$(4+,4+)  \to (4+,4+)$	&  $0$\\
\hline
\end{tabular}
\end{table}


We now compute how adding a single edge $(u,v)$ to $G_i$ decreases the value of
$3\multipl_1(G_i)+2 \multipl_2(G_i) + \multipl_3(G_i)$.
There are ten possibilities for the changes in the degrees of $u$ and $v$ ($\deg_{G_i}(u)$ and $\deg_{G_i}(v)$) that cause $3\multipl_1(G_i)+2 \multipl_2(G_i) + \multipl_3(G_i)$ to change by $\delta$.
We list
the ten possibilities
in the above table where we denote degrees of $4$ or more by $4+$.
It follows that adding a single edge decreases
$3\multipl_1(G_i)+2 \multipl_2(G_i) + \multipl_3(G_i)$ by at most $2$
which implies that $\omegaS(d^i)$ is also decreased by at most $2$.
Also, $2\edgeD(d^i)$ is decreased by exactly $2$ when adding a single edge
to $G_{i-1}$.
Therefore, Inequality~\eqref{goal} holds.
Combining this with inequality~\eqref{base ground} implies that $f(G_0)\geq 0$.
The lemma follows because $f(G_0)=2\edgeD -\omegaS$.
\end{proof}

%
We start with a sequence $d$ satisfying $d_n=2$ and $\multipl_2 = 2$.
To prove our next result, we generate a reduced sequence $d'$ of $d$
where each degree is decreased by $2$.
We verify that $d'$ is forestic and realize it by a graph $G'$ constructed
as the union of a caterpillar graph and a matching,
following \Cref{obs:caterpillar}.
Then, a realization $G$ of $d$ is found by adding a Hamiltonian cycle to $G'$.
The Hamiltonian cycle is the outer-cycle of $G$.
It follows that in this case, there is a realization of $d$ where each vertex
appears exactly once on the outer-cycle.

\begin{lemma}
\label{lem:OP-4n-1}
Let $d\in\cD$.
If
\begin{enumerate}[(i)]
\item
$d_n = 2$ and
\item
$\multipl_2 = 2$,
\end{enumerate}
then $d$ is outerplanaric, and has a realization with an outerplanar
embedding that has a Hamiltonian outer-cycle.
\end{lemma}

\begin{proof}
Let $d$ be as in the lemma. First, we construct the sequence
$d' = (d'_1, \ldots, d'_{n'})$ by subtracting $2$ from each degree of $d$,
i.e., let $d'_i = d_i - 2$ for $i \in \set{1, \ldots, n-2}$.
Note that $n' = n - 2$ and that $\sum d'$ is even.
Moreover, note that $d'$ is forestic, since $\sum d \leq 4n - 6$.
Hence,
\[
\sum d' = \sum d - 2n \leq 2n - 6 = 2(n'+2) - 6 = 2n' - 2.
\]
By \Cref{obs:caterpillar}, $d'$ can be realized by a graph
$G' = (V',E')$ composed of the union of a caterpillar graph $T'$ and
a matching $M'$.
Let $S = (x_1, \ldots, x_s)$ be the vertices on the spine of $T'$, and let
$X_i =  \set{ \ell_{i,1}, \ldots,\ell_{i,k_i} } \subseteq V'$ be the leaves
adjacent to the spine vertex $x_i$, for $i \in \set{1, \ldots, s}$.
Note that vertex $x_i$ has degree $k_i + 1$ if $i \in \set{1,s}$, otherwise
it has degree $k_i + 2$. Assume
$M'=\set{(x'_1,x''_1),\ldots,(x'_{t},x''_{t})}$.
To find an outerplanar realization of $d$, we add a set of edges to $G'$
that form a Hamiltonian cycle (including two additional vertices of degree $2$).
We explain our construction in four steps for which we assume that $s$ is odd.
If $s$ is even, an analogogous construction applies.
\begin{enumerate}[(1)]
\item
Add the edges along the following two paths $P_1 = (x_1,x_3,\ldots,x_{s-2},x_s)$
and $P_2 = (x_{s-1},x_{s-3}, \ldots, x_4,x_2)$ connecting spine vertices in odd and even positions, respectively.
\item
Add two additional vertices $x_0, x_{s+1}$ and use them to form a cycle $C$ together with $P_1$ and $P_2$.
Let $C = (x_0, x_1,x_3, \ldots, x_s, x_{s+1}, x_{s-1}, x_{s-3}, \ldots, x_4,x_2, x_{0})$.
\item
The leaves of a vertex in $P_1$ are inserted into $C$ between two vertices of $P_2$.
Specifically, the leaves of $x_i$, for $i = 1, \ldots, s$,
are inserted between $x_{i-1}$ and $x_{i+1}$.
\item
Insert vertices $x'_1,\ldots,x'_t$ on the segment between $\ell_{s,k_s}$
and $x_{s+1}$ in clockwise order.
Similarly, insert vertices $x''_1,\ldots,x''_t$ on the segment between $x_{s+1}$ and $x_{s}$
in a counter-clockwise order. Connect $x'_i$ and $x''_i$ for $1 \leq i \leq t$.
\end{enumerate}

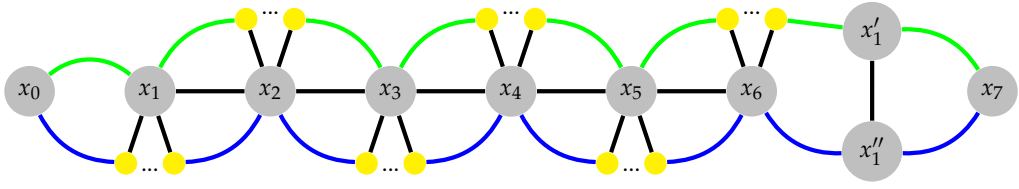
\begin{figure}[t]
\centering
\begin{footnotesize}
	\begin{tikzpicture}[scale=0.8]
		\def\pathlen{6}
		\def\pathspacing{2}
		\def\halfspacing{1}
		\pgfmathtruncatemacro{\plenminusone}{\pathlen - 1}
		\pgfmathtruncatemacro{\plenplusone}{\pathlen + 1}
		\node [minimum size= 3mm,fill=lightgray,circle] (v0) at (0,0)  {$x_{0}$};
		\node [minimum size= 3mm,fill=lightgray,circle] (v\plenplusone) at (\plenplusone * \pathspacing + 2 * \halfspacing,0)  {$x_{\plenplusone}$};
		\foreach \x in {1,...,\pathlen} {
			\pgfmathtruncatemacro{\xcoord}{\x * \pathspacing}
			\node [minimum size= 3mm,fill=lightgray,circle] (v\x) at (\xcoord,0)  {$x_{\x}$};
		}
		\path [draw=green,ultra thick] (v0) [bend left=40] edge (v1);
		\foreach \x in {1,...,\pathlen} {
		        \pgfmathtruncatemacro{\parity}{(-1)^\x}
			\pgfmathtruncatemacro{\xcoord}{\x * \pathspacing}	
			\pgfmathtruncatemacro{\xmm}{\x-1}
			\pgfmathtruncatemacro{\xpp}{\x+1}
			\node [fill=yellow,circle] (leaf1) at ( \xcoord - .4, \parity * 1.2)  {};
			\path [draw=black,ultra thick] (leaf1) [-] edge (v\x);
			\node at ( \xcoord, \parity * 1.3)  {$...$};
			\node [fill=yellow,circle] (leaf2) at ( \xcoord + .4, \parity * 1.2)  {};
			\path [draw=black,ultra thick] (leaf2) [-] edge (v\x);
			\ifthenelse{\parity > 0}{
				\path [draw=green,ultra thick] (leaf1) [bend right] edge (v\xmm);
				\ifthenelse{\x<\pathlen}{
					\path [draw=green,ultra thick] (leaf2) [bend left] edge (v\xpp);
				}{}
			}{
				\path [draw=blue,ultra thick] (leaf1) [bend left] edge (v\xmm);
				\path [draw=blue,ultra thick] (leaf2) [bend right] edge (v\xpp);
				}
	}
		\foreach \x in {2,...,\pathlen} {
			\pgfmathtruncatemacro{\cur}{\x - 1}
			\path [draw=black,ultra thick] (v\cur) [-] edge (v\x);
		}

		\node [minimum size= 3mm,fill=lightgray,circle] (m\pathlen) at (\plenplusone * \pathspacing,1)  {$x'_1$};
		\path [draw=green,ultra thick] (leaf2) edge (m\pathlen);
		\path [draw=green,ultra thick] [bend left] (m\pathlen) edge (v\plenplusone);
		\node [minimum size= 3mm,fill=lightgray,circle] (n\pathlen) at (\plenplusone * \pathspacing,-1)  {$x''_1$};
		\path [draw=blue,ultra thick] [bend right] (v\pathlen) edge (n\pathlen);
		\path [draw=blue,ultra thick] [bend right] (n\pathlen) edge (v\plenplusone);
		\path [draw=black,ultra thick] (m\pathlen) edge (n\pathlen);
	\end{tikzpicture}
\end{footnotesize}
\caption{\small\sf
Illustration of $G$ as constructed in Lemma~\ref{lem:OP-4n-1}.
The vertices $x_1, \ldots,x_6$ and the leaves in yellow (labels omitted) together with the black edges form the caterpillar graph $G'$.
Note that $x_5$ has no leaves.
Paths $P_1$ and $P_2$ are depicted in green and blue, respectively.
Together, they form the cycle $C$.}
\label{fig:caterpillar-fullcycle}
\end{figure}

Figure~\ref{fig:caterpillar-fullcycle} illustrates the construction of $G$.
Following the four steps, we have
\begin{align*}
C &= (x_0,x_1, \ell_{2,1},\ell_{2,2}, \ldots, \ell_{2,k_2}, x_3, \ell_{4,1},\ell_{4,2}, \ldots, \ell_{4,k_2}, x_5, \ldots,
           x_{s-1}, \ell_{s,1},\ell_{s,2}, \ldots, \ell_{s,k_s}, \\
      &  x'_1, \ldots, x'_t, x_{s+1}, x''_t, \ldots, x''_1,
          x_{s}, \ell_{s-1,k_{s-1}}, \ell_{s-1,k_{s-1}-1}, \ldots , x_2, \ell_{1,k_1}, \ell_{1,k_1-1}, \ldots, \ell_{1,1}, x_0)
~.
\end{align*}

%
Formally, let $V = V' \cup \{ x_0,x_{s+1} \}$.
The set $E'$ contains the path connecting the spine vertices,
$\{ (x_i, x_{i+1}) \mid i = 1,2, \ldots, s-1 \}$,
and the edges connecting each spine vertex to its leaves.
The set $L$ contains edges among the leaves while the sets $A_1$ and $A_2$
contain edges connecting spine and leaf vertices
\begin{align*}
L &=~ \{ (\ell_{i,j},\ell_{i,j+1}) \mid i=1,\ldots,s,~~~ j=1,\ldots,k_i-1 \}, \\
A_2 &=~ \{ (x_{i}, \ell_{i+1,1}) \cup (\ell_{i+1, k_{i+1}},x_{i+2}) \; | \; i = 1, \ldots, s \text{ and } k_{i+1} > 0 \},
\\
A_1 &=~ \{ (x_{i}, x_{i+2}) \; | \; i = 1, \ldots, s \text{ and } k_{i+1} = 0 \}
\end{align*}
Recall that $M'=\set{(x'_1,x''_1),\ldots,(x'_{t},x''_{t})}$.
Note that $A_2,A_1,L,E',M'$ are pairwise disjoint.
Let $E = E' \cup L \cup A_2 \cup A_1\cup M'$, and let $G = (V,E)$.
Notice that Figure~\ref{fig:caterpillar-fullcycle}
depicts an outerplanar embedding of $G$.
By the construction of $G$, every vertex in $V$ appears on the Hamiltonian
outer-cycle. No two edges in $E(G)$ cross each other.
Therefore, $G$ is an outerplanar graph.
By the definition of $d'$, two vertices of degree 2 are deleted,
which correspond to $x_0,x_{s+1}$.
Every vertex in  $V\setminus\{x_0,x_{s+1}\}$ has its degree decreased by 2.
Recall that $T'\cup M'$ realizes $d'$. After adding $C$ to $T'$ and inserting
$M'$ in step (4), every vertex $i\in V\setminus\{x_0,x_{s+1}\}$ has its degree
increased by 2, making it equal to $d_i$.
Combining this with the fact that $x_0$ and $x_{s+1}$ have degree 2 in $G$,
we have $\deg(G) = d$.
The lemma follows.
\end{proof}



\begin{corollary}	
\label{lem:OP-4n-2}
$d\in\cD$ is outerplanaric
If $\sum d \leq 4n - 2 \multipl_2 - 2$ and
$d_{n-1} = d_{n} = 2$.
\end{corollary}

\begin{proof}
Let $d''$ be the prefix of $d$ satisfying $d = d'' \circ (2^{\multipl_2})$.
Consider the sequence $d' = d'' \circ (2^2)$ of length $n' = n-\multipl_2+2$.
Observe that $\sum d'$ is even.
As $\sum d \leq 4n - 2 \multipl_2 - 2$, it follows that
\[\sum d' ~=~ \sum d - 2(\multipl_2 - 2) ~\leq~ 4n - 4 \multipl_2 + 2
~=~ 4(n - \multipl_2 + 2) - 6 ~=~ 4n' - 6.\]
Sequence $d'$ satisfies the preconditions of Lemma~\ref{lem:OP-4n-1}.
Hence, let $G'$ be an outerplanar realization and embedding of $d'$
with a Hamiltonian outer-cycle.
An outerplanar realization of $d$ is obtained by replacing an edge on
this outer-cycle
by a path of length $\multipl_2 - 2$.
\end{proof}

\section{Approximate Realizations by Book Embeddings}

This section concerns finding approximate realizations.
Our main result is a classification of $\cD$ into two sub-families,
$\cD=\cD_{NOP}\cup\cD_{2PBE}$, such that
\begin{compactitem}
\item every sequence in $\cD_{NOP}$ is provably non-outerplanaric, and
\item every sequence in $\cD_{2PBE}$ is given a realizing gaph enjoying
  a 2-page book embedding (hereafter referred to as a \emph{2PBE realization}).
\end{compactitem}
Note that if a graph has pagenumber $2$, then it is planar (the converse it false).

\subsection{A 2PBE realization for sequences with $d_1\le 4$}

In this section, we analyze the subclass $\cD_{\le 4}\subseteq \cD$ of sequences
whose largest degree satisfies $d_1\le 4$.
We show that
every $d\in\cD_{\le 4}$ enjoys a 2PBE realization.
Lemmas~\ref{lem:2pbe-d1leq4},
\ref{lem:large dn} and \ref{lem:2pbe-d1leq4 with 1}
imply Theorem~\ref{thm:d1leq4}.


%
%

First, consider the case where $\multipl_3 > 0$ and $d_{n-1} = d_{n} = 2$.

\begin{lemma}
\label{lem:d1leq4-1}
Every $d\in \cD_{\le 4}$ s.t.\
$d_{n} = d_{n-1} = 2$ and $\multipl_3 > 0$ is outerplanaric.
\end{lemma}

\begin{proof}
The degree sum $\sum d$ is even, so necessarily $\multipl_3\geq 2$. Since
$n=\multipl_4 + \multipl_3 + \multipl_2$, it follows that
$\sum d = 4\multipl_4 + 3\multipl_3 + 2\multipl_2 \leq 4n-2\multipl_2-2$,
which satisfies the conditions of
Corollary \ref{lem:OP-4n-2}. Therefore, $d$ is outerpananric.
\end{proof}

%
%


\begin{lemma}
\label{lem:2pbe-d1leq4}
Every $d\in\cD_{\le 4}$ s.t.\ $d_n=d_{n-1}=2$ and $d_1\le n-1$
has an OP+2 realization.
\end{lemma}

\begin{proof}
When $\multipl_3>0$, the claim follows by Lemma
\ref{lem:d1leq4-1}.
So hereafter suppose $\multipl_3=0$.
We split the analysis into three cases.

\begin{description}

\item[Case 1:] $\multipl_4=0$.
Then $d=(2^n)$, which for $n\ge 3$ can be realized by a cycle,
hence it is outerplanaric.

\item[Case 2:] $\multipl_4=1$.
Then $\multipl_2\geq 4$ since $\multipl_3=0$ and $d_1=4\le n-1$.
In this case, construct an outerplanar graph $G$ as follows.
First, connect vertex $v_1$ of degree 4 to vertices $v_2,v_3,v_4,v_5$
of degree 2. Next, add edges $(v_2,v_3), (v_4,v_5)$ and insert $\multipl_2-4$
vertices into the edge $(v_2,v_3)$, i.e., transform this edge into a path
$(v_2,u_1,\ldots,u_{\multipl_2-4},v_3)$. Notice that $G$ is outerplanar.

\item [Case 3:] $\multipl_4\geq 2$.
Here, $d$ consists of only 2's and 4's.
Since $\sum d\leq 4n-6$, necessarily
\begin{align}\label{number of degree 2}
\multipl_2\geq 3.
\end{align}
Generate a sequence $d'$ by removing one degree $2$ from $d$ and reducing
$d_1, d_2$ by $1$. Then $d'=(4^{\multipl_4-2},3^2,2^{\multipl_2-1})$.
It is easy to verify, using Eq. \eqref{number of degree 2},
that $d'$ satisfies the preconditions of Lemma \ref{lem:d1leq4-1}.
Hence, let $G'$ be an outerplanar realization of $d'$ by Lemma
\ref{lem:d1leq4-1}.
We create a 2-page book embedding graph $G$ realizing $d$ based on
$G' = (V',E')$.
Let $v_1,v_2 \in V'$ be the vertices with $\deg_{G'}(v_1)= \deg_{G'}(v_2) =3$.
Add an isolated vertex $u$ to $G'$ and obtain an outerplanar graph
$G_1= (V,E')$ with $V=V'\cup \{u\}$.
Let $G_2= (V,E_2)$ with $E_2=\{(v_1,u),(v_2,u)\}$.
The position of $u$ in the embeddings of $G'$ and $G_2$ can be fixed
arbitrarily since $u$ is isolated in $G_1$, and $G_2$ has only two edges
that do not exist in $\bar{G}'$.
Hence, $G_1$ and $G_2$ form
a 2PBE realization for $d$.
Notice that $G_1$ is outerplanar and $G_2$ consists of two edges.
\end{description}
The lemma follows.
\end{proof}

\begin{lemma}
\label{lem:large dn}
Every $d\in \cD$ s.t.\
$d_n\ge 2$ and $d_{n-1}\ge 3$
has
an OP+1 realization.
\end{lemma}

\begin{proof}
Since $\sum d\leq 4n-6$ and $d_n\geq 2$, we have $d_{n-1}\leq 3$.
Consider the two possible cases, $(d_{n-1}, d_{n}) \in\{(2,3),(3,3)\}$.

\begin{description}
\item[Case 1:] $d_{n}=d_{n-1}=3$.
Then construct a sequence $d'$ by replacing $(d_{n-1},d_{n})$ with $(2,2)$.
Note that $d'$ satisfies the conditions in Lemma \ref{lem:OP-4n-1} and it can
be realized by a outerplanar graph $G'$ with two nonadjacent vertices
$v',v''$ of degree 2. Then $d$ can be realized by a graph $G$ obtained by
adding the edge $(v',v'')$ to $G'$. One can verify that $G$ has a
2-page book embedding (with $G'$ on one page and the added edge on the other).

\item[Case 2:]
$d_{n}=2$ and $d_{n-1}=3$.
In this case, $d_{n-2}=3$ since $\sum d\leq 4n-6$. Construct $d'$ by replacing
$(d_{n-2}, d_{n-1})$ with $(2,2)$. Note that $\sum d'\leq 4n-8$ and
$\multipl'_2=3$, so $d'$ satisfies the conditions in
Corollary \ref{lem:OP-4n-2} and there is an outerplanar graph $G'$
realizing $d'$. By the construction of $G'$, it has two nonadjacent vertices
$v',v''$ of degree 2. Therefore, we can construct a graph $G$ by
adding the edge $(v',v'')$ to $G'$. Again, $G$ has a 2-page book embedding.
\end{description}
Notice that in both cases $G'$ is outerplanar and the other page consists of one edge. The lemma follows.
\end{proof}

It remains to deal with sequences $d\in\cD_{\le 4}$ that contain also 1's.
We show the following.

\begin{lemma}
\label{lem:2pbe-d1leq4 with 1}
Every $d\in\cD_{\le 4}$ s.t.\
$d_n=1$ and $d_1\le n-1$
has
an OP+2 realization.
\end{lemma}

\begin{proof}
First, in $\multipl_1$ steps,
delete one degree 1 from the sequence and decrease the current largest degree by 1,
yielding the residual sequence $d'$, whose length is $n' = n - \multipl_1$.
Since $\sum d\leq 4n-2\multipl_1-6$ and $\sum d\geq 2n$,
\begin{align}
\label{bound of sumd and dn}
\sum d' & \leq 4(n-\multipl_1)-6, & d'_{n-\multipl_1} & \geq 2
\end{align}
Consider the following two cases.

\begin{description}
\item[Case 1:]
$\multipl_1\leq \multipl_4$. Then
$d'=(4^{\multipl_4-\multipl_1},3^{\multipl_3+\multipl_1},2^{\multipl_2})$.
and that when $n'=4$, $d'$ is $OP+1$.
In this case, if $\multipl_2\geq 2$ (respectively, $\multipl_2\leq 1$),
then $d'$ has a 2PBE realization $G'$
by Eq. \eqref{bound of sumd and dn} and
Lemma \ref{lem:2pbe-d1leq4},
(resp., \ref{lem:large dn}).
Notice that for $\multipl_2\ge 2$, $d'_1\le n'-1$ since $\multipl_3+\multipl_1\geq 2$.

\item[Case 2:]
$\multipl_1>\multipl_4$. Then
$$d' = (3^{\multipl_3 + \multipl_4 - (\multipl_1 -\multipl_4)}, 2^{\multipl_2 + (\multipl_1-\multipl_4)})
= (3^{\multipl_3+2\multipl_4-\multipl_1},2^{\multipl_2+\multipl_1-\multipl_4}),$$
which satisfies the condition in Lemma \ref{lem:large dn} or
\ref{lem:2pbe-d1leq4},
by Eq. \eqref{bound of sumd and dn}, so again
$d'$ has a 2PBE realization $G'$.
Notice that for $\multipl'_2=\multipl_2+\multipl_1-\multipl_4\ge 2$, $d'_1\le n'-1$.
\end{description}

Finally, we obtain a 2PBE realization $G$ for $d$
by taking $G'$ and connecting $\multipl_1$ new degree 1 vertices with
degree 3 or 2 in $G'$.
By above construction of $G$, one page is outerplanar and the other page
consists of at most two edges.
The lemma follows.
\end{proof}

Combining Lemmas~\ref{lem:2pbe-d1leq4},
\ref{lem:large dn} and \ref{lem:2pbe-d1leq4 with 1},
we have the following theorem.

\begin{theorem}
\label{thm:d1leq4}
Let $d\in\cD_{\le 4}$ and $d_1\le n-1$.
Then $d$ has
an OP+2 realization.
\end{theorem}

\subsection{A 2PBE realization for sequences with $d_1\geq 5$}

In this section, we analyze the subclass $\cD_{\ge 5} = \cD \setminus \cD_{\le 4}$
of sequences whose largest degree satisfies $d_1\ge 5$.
We first consider the special case of \emph{maximal-volume} sequences
in $\cD_{\ge 5}$, namely, ones whose volume equals exactly $4n-6$.
By \Cref{lem:outerplanar-sum-bound-omega1}, $\multipl_1 = 0$.
Combining this with \Cref{lem:op_nec_deg23}, $\multipl_2 \geq 2$.
Since \Cref{lem:OP-4n-1} has already considered the case for
$\multipl_2=2$, it remains to consider sequences $d$ where $\multipl_2 > 2$.
This class of sequences turns out to be the most challenging.
Their characterization is covered by the next key lemma, whose proof
is by far the most complex part of our analysis.

\begin{lemma}
\label{lem:BE}
Let $d\in\cD_{\ge 5}$.
If
\begin{inparaenum}[(1)]
\item $\sum d = 4n - 6$, \label{eq:sumD}
\item $2 \multipl_2 + \multipl_3 \leq n + 1$,
\item $\multipl_2 > 2$, and  \label{assumption:omega-two}
\item $d_n = 2$,
\end{inparaenum}
Then $d$ has
an OP+bi realization.
\end{lemma}

\begin{figure}[ht]
\centering
\begin{footnotesize}
\begin{tikzpicture}[scale=0.4]	
\def\pathlen{5}
\def\pathspacing{5}
\pgfmathtruncatemacro{\plenplusone}{\pathlen + 1}
			
\foreach \x in {1,...,\pathlen} {
				
	\pgfmathtruncatemacro{\xcoord}{\x * \pathspacing}
	\node [label={[below = 8mm]$v_{\pi(\x)}$},minimum size= 4mm,fill=cyan] (v\x) at (\xcoord,0)  {};
			}
\foreach \x in {1,...,4} {
				
	\pgfmathtruncatemacro{\xcoord}{2 + \x *  \pathspacing}
	\node [label={[below = 8mm]$u_{\x}$},minimum size= 4mm,fill=red,circle] (u\x) at (\xcoord,0)  {};
			}
		
	\path [draw=black,ultra thick] (v1) [bend left=30] edge (u1);	
	\path [draw=black,ultra thick] (u1) [bend left=30] edge (v2);
	\path [draw=black,ultra thick] (v2) [bend left=30] edge (u2);		
	\path [draw=black,ultra thick] (u2) [bend left=30] edge (v3);
	\path [draw=black,ultra thick] (v3) [bend left=30] edge (u3);
	\path [draw=black,ultra thick] (u3) [bend left=40] edge (v5);
	
	\path [draw=black,ultra thick] (v4) [bend left=30] edge (u4);
	\path [draw=black,ultra thick] (u4) [bend left=30] edge (v5);

\end{tikzpicture}
\end{footnotesize}
\caption{\small\sf
Illustration for the way in which the graph $G'$ constructed in the proof
of Case A of Lemma~\ref{lem:BE} can be embedded into the page of a book.
The blue rectangles are B's vertices. In the example, $j = 4$.}
\label{fig:path-book-spine}
\end{figure}
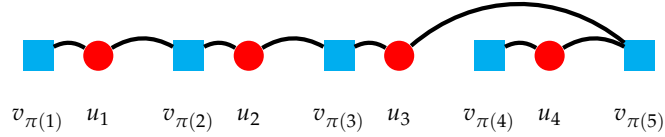

\begin{figure}[ht]
\begin{center}
\includegraphics[width=2in]{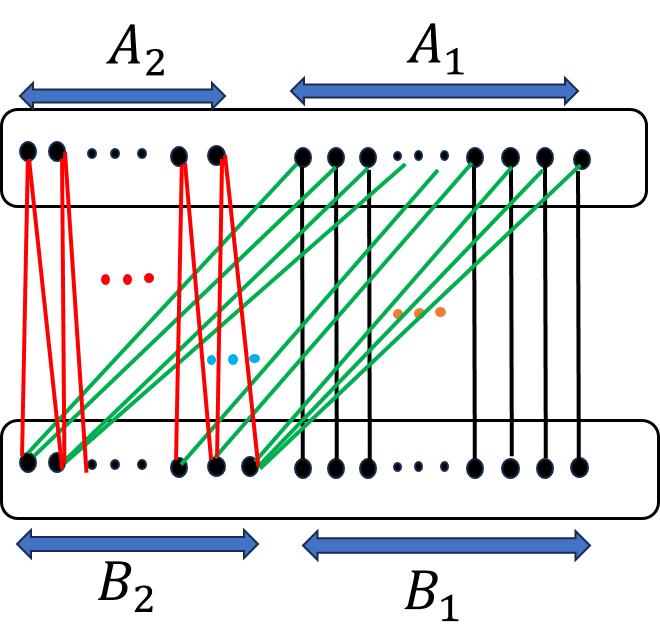}
\end{center}
\caption{\small\sf Illustration of the construction of $Bar(A,B)$.
The black edges represent the matching between $A_1$ and $B_1$.
The green edges represent the union of disjoint stars between $B_2$ and $A_1$.
The red edges represent the zigzag path between $A_2$ and $B_2$}
\label{fig:bipartite-graph}
\end{figure}

\begin{figure}[ht]
\centering
\begin{footnotesize}
	\begin{tikzpicture}[scale=.7]
		\node [label={[below = 7mm]$v_{\pi(1)}$},minimum size= 4mm,fill=cyan] (v1) at (0.5,0)  {};
		\node [label={[below = 7mm]$v_{\pi(2)}$},minimum size= 4mm,fill=cyan,circle] (v1) at (2,0)  {};
		\node [label={[below = 7mm]$a_2$},minimum size= 4mm,fill=red,circle] (s1) at (2.8,0)  {};
		\path [draw=black,ultra thick] (v1) [bend left=30] edge (s1);
		\node [label={[below = 7mm]$v_{\pi(3)}$},minimum size= 4mm,fill=cyan,circle] (v1) at (4,0)  {};
		\node [label={[below = 7mm]$a_3$},minimum size= 4mm,fill=red,circle] (s1) at (4.8,0)  {};
		\path [draw=black,ultra thick] (v1) [bend left=30] edge (s1);
		\node [label={[below = 7mm]$v_{\pi(4)}$},minimum size= 4mm,fill=cyan] (v1) at (6,0)  {};
		\node [label={[below = 7mm]$v_{\pi(5)}$},minimum size= 4mm,fill=cyan,circle] (v1) at (7.4,0)  {};
		\node [label={[below = 7mm]$a_5$},minimum size= 4mm,fill=red,circle] (s1) at (8.2,0)  {};
		\path [draw=black,ultra thick] (v1) [bend left=30] edge (s1);
		\node [label={[below = 7mm]$v_{\pi(6)}$},minimum size= 4mm,fill=cyan] (v1) at (9.4,0)  {};
		\node [label={[below = 7mm]$v_{\pi(7)}$},minimum size= 4mm,fill=cyan] (v1) at (10.8,0)  {};
		\node [label={[below = 7mm]$v_{\pi(8)}$},minimum size= 4mm,fill=cyan,circle] (v1) at (12.2,0)  {};
		\node [label={[below = 7mm]$a_8$},minimum size= 4mm,fill=red,circle] (s1) at (13.0,0)  {};
		\path [draw=black,ultra thick] (v1) [bend left=30] edge (s1);
		\node [label={[below = 7mm]$v_{\pi(9)}$},minimum size= 4mm,fill=cyan,circle] (v1) at (14.4,0)  {};
		\node [label={[below = 7mm]$a_9$},minimum size= 4mm,fill=red,circle] (s1) at (15.4,0)  {};
		\path [draw=black,ultra thick] (v1) [bend left=30] edge (s1);
	\end{tikzpicture}
\end{footnotesize}
\caption{\small\sf
The graph $G'$ after the first step of the construction.
Throughout our illustration figures,
vertices of $A_1$ (respectively, $B$) are colored red (resp., blue).
Vertices of $A_1$, $B_1$ (resp., $A_2$, $B_2$) are circular (resp., rectangular).}
\label{fig:book-emb1}
\end{figure}
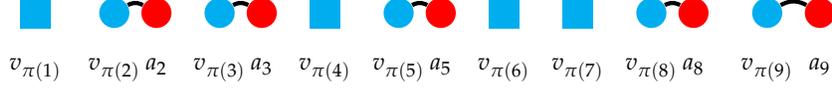

\begin{figure}[h]
\begin{center}
\includegraphics[width=5in]{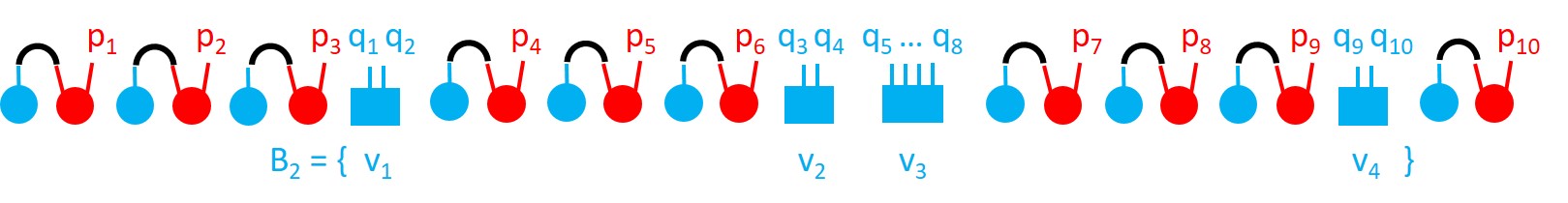}
\end{center}
\caption{\small\sf The situation before Step 2.
Vertices of $A$ (respectively, $B$) are colored red (resp., blue).
Vertices of $A_1$ and $B_1$ (respectively, $B_2$) are depicted as circles
(resp., rectangles).
The free ports of a vertex are depicted as tags attached to it.
In this example, the $d'$ degrees of the vertices of $B_2$ are
$d'(v_1)=3$, $d'(v_2)=4$, $d'(v_2)=6$ and $d'(v_4)=3$.
The ports of $B_2$ and $A_1$ vertices that have to be connected
in a non-intersecting way correspond to the input sequences of the example
in Figure \ref{fig:bipartite-balanced-brackets}.
This is done using Algorithm \BBB\ of Section \ref{sec:2}.
}
\label{fig:connecting-B2-A1}
\end{figure}

\begin{figure}[h]
\begin{center}
\includegraphics[width=5in]{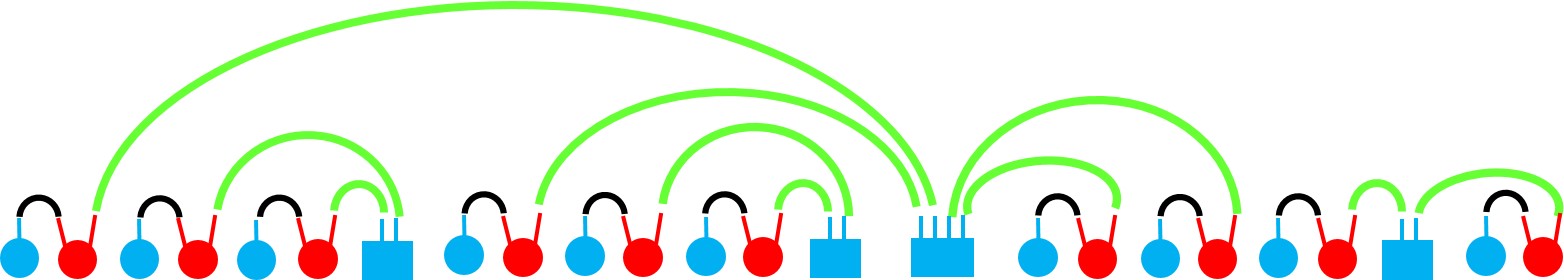}
\end{center}
\caption{\small\sf The connections obtained by Algorithm \BBB\ for the ports
of $B_2$ and $A_1$.}
\label{fig:connecting-B2-A1-solution}
\end{figure}

\begin{figure}[h]
\begin{center}
\includegraphics[width=5in]{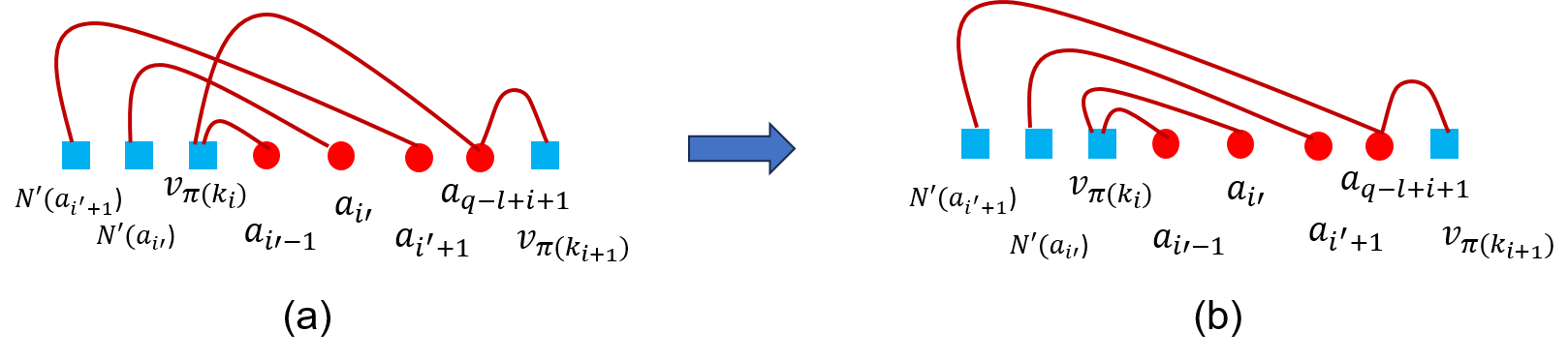}
\end{center}
\caption{\small\sf Illustration of Algorithm \ref{alg:correction},
which transforms the crossing edges to be noncrossing.
(a) The graph obtained after placing the vertices of $A_2$ in the first part
of step 3. In this graph, the edge $(v_{\pi(k_i)},a_{q-\ell+i+1})$ crosses
the edges $(N'(a_{i'+1}),a_{i'+1})$ and $(N'(a_{i'}),a_{i'})$.
(b) The resulting graph after the transformation. The three edges involvesd
in the crossings are replaced by the edges
$(v_{\pi(k_i)},a_{i'}), (N'(a_{i'+1}),a_{q-\ell+i+1}), (N'(a_{i'}),a_{i'+1})$.}
\label{fig:correction}
\end{figure}

\begin{figure}[h]
\begin{center}
\includegraphics[width=5in]{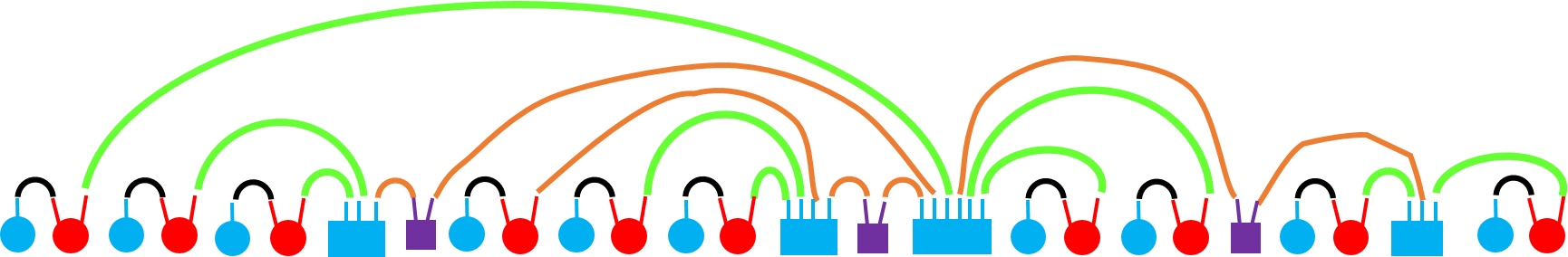}
\end{center}
\caption{\small\sf The graph after Algorithm \ref{alg:correction} based on
Figure \ref{fig:connecting-B2-A1-solution}. The purple squares represent
the vertices in $A_2$. The orange edges represent the modified edges based on
Figure \ref{fig:connecting-B2-A1-solution}.}
\label{fig:final-graph}
\end{figure}

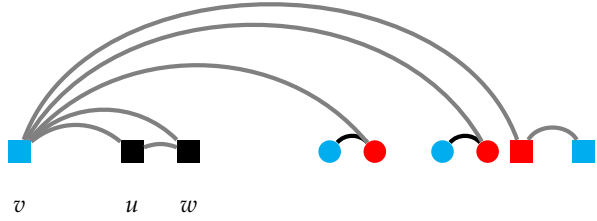
\begin{figure}[h]
	\centering
	\begin{footnotesize}
		\begin{tikzpicture}[scale=.75]
			
			\node [label={[below = 7mm]$v$},label={[above = 4mm]$$},minimum size= 3mm,fill=cyan] (v1) at (-4,0)  {};
			
			\node [label={[below = 7mm]$u$},label={[above = 4mm]$$},minimum size= 3mm,fill=black] (u) at (-2,0)  {};
			
			\node [label={[below = 7mm]$w$},label={[above = 4mm]$$},minimum size= 3mm,fill=black] (w) at (-1,0)  {};
			\path [draw=gray,ultra thick] (v1) [bend left=40] edge (u);
			\path [draw=gray,ultra thick] (v1) [bend left=45] edge (w);
			\path [draw=gray,ultra thick] (u) [bend left=20] edge (w);
			
			\node [label={[below = 7mm]$$},minimum size= 3mm,fill=cyan,circle] (v2) at (1.5,0)  {};
			\node [label={[below = 7mm]$$},minimum size= 3mm,fill=red,circle] (s1) at (2.3,0)  {};
			\path [draw=black,ultra thick] (v2) [bend left=50] edge (s1);

			\node [label={[below = 7mm]$$},minimum size= 3mm,fill=cyan,circle] (v3) at (3.5,0)  {};
			\node [label={[below = 7mm]$$},minimum size= 3mm,fill=red,circle] (s2) at (4.3,0)  {};
			\path [draw=black,ultra thick] (v3) [bend left=50] edge (s2);
			
			\node [label={[below = 7mm]$$},minimum size= 3mm,fill=red] (u1) at (4.9,0)  {};
			\path [draw=gray,ultra thick] (v1) [bend left=70,fill=green] edge (u1);
			\path [draw=gray,ultra thick] (v1) [bend left=50] edge (s1);
			\path [draw=gray,ultra thick] (v1) [bend left=60] edge (s2);

			\node [label={[below = 5mm]$$},label={[above = 4mm]$$},minimum size= 3mm,fill=cyan] (v4) at (6,0)  {};
			\path [draw=gray,ultra thick] (u1) [bend left=60] edge (v4);
			
		\end{tikzpicture}
	\end{footnotesize}
	\caption{\small\sf The black squares represent the added vertices $u,w$ and they are connected to $v$.}
\label{fig:book-emb3}
\end{figure}

\begin{proof}
Let $d$ be as in the lemma.
To prove the claim, we define sequences $d', d''$,
of lengths $n', n''$ respectively, such that $d = d' \oplus d''$.

Intuitively, it would appear that if property (1) of the lemma holds,
then $d'$ and $d''$ can be defined such that
$\sum d' = 2n'-4$ and $\sum d'' = 2n''-2$, so both $d'$ and $d''$
can be realized by a forest and embedded on its own page,
implying a 2-page book embedding.
Proving this intuition is complicated by the need to ensure two properties:
\begin{inparaenum}[(a)]
\item that no edge $(u,w)$ occurs in \emph{both} pages, and
\item that the vertices occur in the same order on the outer face of both pages.
\end{inparaenum}

Instead, we select the sequence $d''$ (ignoring the $0$-degrees) so that it
meets the preconditions of \Cref{lem:OP-4n-1}, i.e., $d''$ satisfies
\begin{inparaenum}[(i)]
\item \label{item:omega2}
        $\multipl''_2 = 2$,
\item \label{item:dn2}
        $d''_{n''} = 2$,  and
\item $d" \in \cD$, namely, $\sum d'' \leq 4 n'' - 6$ and
$\sum d''$ is even.
\label{item:sum}
\end{inparaenum}
Then, Lemma~\ref{lem:OP-4n-1} provides an outerplanar realization $G''$ of $d''$
and its embedding fills the first page of the book.
The embedding of $G''$ imposes an ordering on the vertices,
as they appear along the spine of the book
(which is the same ordering as on the outer-cycle of the embedding).
For the second page, we construct an outerplanar realization $G'$ of $d'$
such that the vertices follow the same ordering along the book spine.
To obtain $d''$, we apply the following two transformations to $d$.
\begin{description}
	\item[(TR1)] remove $q := \multipl_2 - 2$ degree-2 entries from $d$ to meet requirement~\eqref{item:omega2}.
	\item[(TR2)] decrease the total volume of the remaining sequence
	  by $2 q$ to satisfy requirement~\eqref{item:sum},
          while carefully avoiding the creation of new degrees $1$ or $2$.
\end{description}
Denote the number of degrees that are $5$ or more by
$\ell := $$\sum_{i \geq 5} \multipl_i = n - \multipl_2 - \multipl_3 - \multipl_4$.
The remaining proof, describing the construction of $G'$ for $d'$,
is split into two cases, depending on the relation between $q$ and $\ell$.

\paragraph{Case A:}
$q < \ell$.
In this easier case, a path sequence $d' = (2^{n'-2},1^2)$ is a suitable choice,
and the construction of $G'$ is simplified considerably.

%
We define the sequence $d'$ to have length $n'=n$.
It contains vertices of degree $0$ and we do \emph{not} consider it in
non-increasing order, to satisfy the constraint $d = d' \oplus d''$.
Once $d'$ is defined, $d''$ is implicitly defined as well,
as $d''=d\ominus d'$. It can be regarded as having length $n''=n$
where $d''_i = 0$ if $i > n - q$.
\begin{align*}
d_i' & =
\begin{cases}
1, & \text{ if } i = 1, \\
2, & \text{ if } i = 2, \ldots, q, \\
1, & \text{ if } i = q + 1, \\
0, & \text{ if } i = q+2, \ldots, n - q, \\
2, & \text{ if } i = n - q + 1, \ldots, n.
\end{cases}
&
d_i'' & =
\begin{cases}
d_i - 1, & \text{ if } i = 1, \\
d_i - 2, & \text{ if } i = 2, \ldots, q, \\
d_i - 1, & \text{ if } i = q + 1, \\
d_i, & \text{ if } i = q+2, \ldots, n - q, \\
0, & \text{ if } i = n - q + 1, \ldots, n.
\end{cases}
\end{align*}

We use the sequence $d=(6^2,5^4,4,3^2,2^6)$ as a running example for Case~A.
In this case, $n=15$, $q=4$, and $\ell=6$.
It is not hard to verify that $d$ satisfies the conditions of this lemma and the condition of Case~A.
We have that $d'=(1,2^3,1,0^2,2^4)$, $d''=(5,4,3^2,4,5,4,3^2,2^2,0^4)$.

%
%
For the next step, consider $\hat d = (d''_1,d''_2, \ldots, d''_{n-q})$,
which is $d''$ without the trailing $0$ degrees. In the above example,
we have that
$\hat{d}=(5^2,4^3,3^4,2^2)$.
For this $\hat d$, whose length is $\hat{n} = n-q$, we show the following.

\begin{claim}
\label{cl:in part A}
Sequence $\hat d$ satisfies the preconditions of Lemma~\ref{lem:OP-4n-1}.
\end{claim}

\begin{proof}
By construction of $d''$, we have that $\sum \hat d = \sum d"$ is even,
and that $\hat{d}_{\hat{n}} = d''_{n -q} = 2$.
Observe that $\hat{\multipl}_2 = \multipl''_2 = 2$ as $q < \ell$ implies that
$d_{q+1} \geq 5$.
To complete the proof,
observe that
\[
\sum \hat d = \sum d - 4 q \leq
4n - 6 - 4q = 4(n-q) - 6 = 4\hat{n} - 6. \mbox{\hspace{2cm}}
\Box
\]
\end{proof}

By \Cref{lem:OP-4n-1}, $\hat{d}$ can be realized as the outerplanar graph
$\hat{G}$. (For our example, this graph is given in \Cref{fig:example-case-A}.)
We add $q$ isolated vertices to $V(\hat G)$ to generate an outerplanar
realization $G'' = (V,E'')$ of $d''$.
Let $V = \{ v_1, v_2, \ldots, v_n \}$ such that $\deg_{G''}(v_i) = d''_i$,
for $i = 1,2,\ldots,n$.
Hence, the set $R := \{ v_{n - q + 1}, \ldots, v_{n} \}$ contains the isolated vertices of degree $0$.
Consider an outerplanar embedding of $\hat G$, and note that the isolated
vertices in $R$ can be placed in arbitrary positions
on the cyclic order to yield a planar embedding of $G''$.

%
%
\begin{figure}[t]
\centering
\begin{footnotesize}
\begin{tikzpicture}[scale=0.3]

\begin{scope}[
every node/.style={fill=blue,circle,inner sep=2pt},
]

\node(Pi1) at (0,0) {};
\node(Pi2) at (-3,4) {};
\node(Pi3) at (0,4) {};
\node(Pi4) at (9,4) {};
\node(Pi5) at (6,0) {};

\end{scope}

\begin{scope}[
every node/.style={fill=green,circle,inner sep=2pt},
]

\node(X1) at (-3,0) {};
\node(X2) at (3,0) {};
\node(X3) at (3,4) {};
\node(X4) at (12,0) {};
\node(X5) at (12,4) {};
\node(X6) at (15,4) {};

\path (Pi1) [-] edge (Pi2);
\path (Pi1) [-] edge (Pi3);
\path (Pi2) [-] edge (Pi3);
\path (Pi4) [-] edge (Pi5);

\end{scope}

\begin{scope}[
every edge/.style={draw=black,thick}
]

\path (Pi1) [-] edge (Pi2);
\path (Pi1) [-] edge (Pi3);
\path (X3) [-] edge (Pi1);
\path (X2) [-] edge (X3);
\path (X3) [-] edge (Pi5);
\path (Pi4) [-] edge (Pi5);
\path (X4) [-] edge (Pi4);
\path (X4) [-] edge (X5);

\end{scope}

\begin{scope}[
every edge/.style={draw=red,thick}
]

\path (X1) [-] edge (Pi1);
\path (X2) [-] edge (Pi1);
\path (X2) [-] edge (Pi5);
\path (X4) [-] edge (Pi5);
\path (X4) [-] edge (X6);
\path (X5) [-] edge (X6);
\path (X5) [-] edge (Pi4);
\path (X3) [-] edge (Pi4);
\path (X3) [-] edge (Pi3);
\path (Pi2) [-] edge (Pi3);
\path (X1) [-] edge (Pi2);

\end{scope}

\node () at (0,-1.2) {$v_{\pi(1)}$};
\node () at (-3.1,5.2) {$v_{\pi(2)}$};
\node () at (0.3,5.2) {$v_{\pi(3)}$};
\node () at (9,5.2) {$v_{\pi(4)}$};
\node () at (6,-1.2) {$v_{\pi(5)}$};

\end{tikzpicture}
\end{footnotesize}
\caption{\small\sf The graph realizing $\hat{d}$.
Blue circles: vertices in $B$. Green circles: vertices in $V(\hat{G})\setminus B$.
}
\label{fig:example-case-A}
\end{figure}
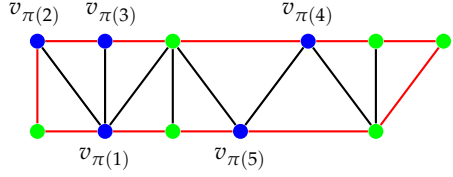
%
%

Let $B := \{ v_1,v_2, \ldots,v_{q+1} \}$ be the set of vertices whose degrees
were reduced in $d''$ compared to the original $d$ and whose original degree is at least 5.
Let $\pi$ be the permutation by which the vertices of $B$ appear
in the cyclic-order defined by the planar embedding of $G''$.
Without loss of generality, let $\pi(1) = 1$.
Hence, there is an embedding of $G''$ into the page of a book where
the vertices in $B$ are arranged along the spine
in the order $(v_{\pi(1)}, v_{\pi(2)}, \ldots, v_{\pi(q+1)})$.
The vertices in $R$ can be placed arbitrarily along the spine.

In our example
$\hat{G}$, $B=\{v_1,\ldots,v_5\}$, and its vertices are arranged
in the outer-cycle in the order
$v_{\pi_1},\ldots,v_{\pi_5}$, as depicted in \Cref{fig:example-case-A}.

To finish the proof of case A, construct a realization $G'=(V,E')$
of $d'$ such that
$\deg_{G'}(v_i) = d'_i$, for $i = 1,2,\ldots,n$.
The vertices $v_{q+2}, \ldots, v_{n - q}$ have vertices of degree $0$ in $G'$
and are ignored when constructing $G'$.
The vertices in $R$ must have degree $2$ in $G'$.
%
In the embeddings of $G''$ and $G'$,
we place vertex $v_{n - q + i}$ between $v_{\pi(i)}$ and $v_{\pi(i+1)}$
along the spine, for $i  = 1, \ldots, q$.
Let index $j$ be such that $d''_{\pi(j)} = d_{q+1} - 1$.
For every $i = 1,\ldots$, $\pi(j)-2$, $\pi(j)$, $\pi(j)+1, \ldots$, $q$,
add to $G'$ the two edges
$(v_{\pi(i)},v_{n - q + i})$ and $(v_{n - q + i}, v_{\pi(i+1)})$.
Finally, add the edges
$(v_{\pi(j)-1}, v_{n - q + j})$ and $(v_{n - q + j}, v_{q+1})$.

By construction, $V(G'') = V(G')$.
By definition of $d'$ and $d''$,
it follows that
$\deg_{G'}(v_i) + \deg_{G''}(v_i) = d_i$ for $i = 1,2, \ldots, n$.
The vertices in $B \cup R$ form a path in $G'$ where $v_1$ and $v_{\pi(j)}$
are the two ends (vertices of degree $1$).
Finally, observe that $E'' \cap E' = \emptyset$
since each edge of $G'$ is incident to a vertex in $R$, which has vertices of
degree 0 in $G''$.
It follows that $G''$ and $G'$ form a valid book embedding
of a graph that realizes $d$. Notice that $G'$ is bipartite outerplanar
and $G''$ is an outerplanar graph.

Notice that in our example $v_{\pi(1)}=v_1$, $v_{\pi(2)}=v_3$, $v_{\pi(3)}=v_4$, $v_{\pi(4)}=v_2$, $v_{\pi(5)}=v_5$.
Then in $G'$, $d'(v_{\pi(1)})=1$, $d'(v_{\pi(2)})=d'(v_{\pi(3)})=d'(v_{\pi(4)})=2$, $d'(v_{\pi(5)})=1$, $R=\{u_1,\ldots,u_q\}$.
Figure~\ref{fig:path-book-spine} shows a valid embedding of $G'$
into the page of a book.
This completes the proof of Case A.

\paragraph{Case B:}
$q \geq \ell$.
This case is significantly more complicated than case A.
Note that as $2 \multipl_2 + \multipl_3 \leq n+1$,
it follows that
\[\multipl_4 = n - \multipl_2 - \multipl_3 - \ell \geq \multipl_2 - 1 - \ell = q+1 - \ell > 0~.\]
Also note that $w_4 + \ell \geq q+1$.

The construction generalizes that of case A.
Similar to case A, we
obtain sequence $d''$ from $d$ by removing
$q=\multipl_2-2$ vertices of degree $2$.
The next steps, though, are more complicated than in case A.
The reason is that,
if we set (like above) $d''_i = d_i - 2$, for $i = 2,3, \ldots, q$,
then some vertices of degree $4$ in $d$ are decreased to $2$ in $d''$
and requirement~\eqref{item:omega2}, 
$\multipl''_2 = 2$, might not be met.

Instead, the sequences $d'$ and $d''$ have to be defined differently.
In particular,
to define $d''$, we reduce only the degrees that are \emph{larger} than $4$
by $2$ (or more), and we reduce $q - \ell + 1$ vertices of degree $4$ to $3$.
At first, let sequence $\bar d$ be such that
\[
\bar{d}_i =
\begin{cases}
  d_i - 2, & \text{ if } i = 1, \ldots, \ell, ~~~~~~~~~~~~~~~~~(\text{where } d_i \geq 5)
  \\
  d_i - 1, & \text{ if } i = \ell + 1, \ldots, q+1 ~~~~~(\text{where } d_i = 4)
  \\
  d_i, & \text{ if } i = q+2, \ldots, n - q, ~~~(\text{where } 4 \geq d_i \geq 2)
  \\
	0, & \text{ if } i = n - q + 1, \ldots, n. ~~~ (\text{where } d_i = 2)
\end{cases}
\]
Recall that $d$ satisfies $w_4 + \ell \geq q+1$,
which implies
that $d_{q+1} \geq 4$ and therefore $\bar{d}_{q+1} \geq 3$.
As $q = \multipl_2 - 2$ out of the $\multipl_2$ vertices of degree $2$
of $d$ are transformed to $0$, it follows that $\bar{\multipl}_2 = 2$.

Note that $\sum \bar d \leq 4(n-q) - 6$ is not guaranteed,
so requirement~\eqref{item:sum} does not necessarily hold.
However, observe that
$\sum \bar d = \sum d - 2(\ell+q) - (q -\ell + 1)$, so

\begin{equation}
\label{eq:decrease-volum}
\sum \bar d
=    \sum d - 2(\ell+q) - (q -\ell + 1)
=    4n - 6 - 3q - \ell - 1
=    4(n-q) - 6 + (q - \ell - 1)
\end{equation}

%
does hold.
To meet requirement~\eqref{item:sum} of \Cref{lem:OP-4n-1}, we set $d'' = \bar d$ and
gradually reduce $\sum_{i=1}^{\ell} d''_i$ further by decreasing the large degrees using Algorithm~\ref{alg:dprimeprime}.

\begin{algorithm}[h]
\caption{\small\sf Generating sequence $d''$ from $\bar d$.}
\label{alg:dprimeprime}
	Set $d'' \leftarrow \bar d$ \\
	\While {$\sum d'' > 4(n-q) - 6$}{
		$d''_j \gets d''_j - 1$, where $j = \text{argmax}_{i \in [\ell]} ~ d''_i$\\
	}
        $\hat d \gets\pos(d'')$.
\end{algorithm}

Note that the output $\hat d = (d''_1, d''_2, \ldots, d''_{n-q})$ of
Algorithm~\ref{alg:dprimeprime} is the prefix of length ${\hat n} = n-q$
of the sequence $d''$ at the end of
Algorithm~\ref{alg:dprimeprime}, omitting the trailing zeros.

\begin{claim}
\label{claim: hatd satisfies preconditions}
The sequence $\hat d$ satisfies the preconditions of \Cref{lem:OP-4n-1}.
\end{claim}

\begin{proof}
Requirement~\eqref{item:sum} holds by construction.
%
Note that when constructing $\bar d$ we reduced $2\ell$ from the volume
of the ``high'' (first $\ell$) elements, $\sum_{i=1}^{\ell} d_i$.
Algorithm~\ref{alg:dprimeprime} reduces this volume further by $(q - \ell - 1)$.
To show that $\hat d$ satisfies requirements \eqref{item:omega2} and
\eqref{item:dn2},
it suffices to show that the ``high'' degrees $d_1, d_2, \ldots, d_{\ell}$
have sufficient volume so that $d''_i \geq 3$ for $i \leq \ell$, i.e.,
it remains to prove that $\sum_{i = 1}^{\ell} (\bar{d}_i - 3) \geq q - \ell - 1$.

Recall that $\sum d = 4n -6$, which can be rewritten as
\begin{equation}
  \label{eq:in claim}
\sum_{i = 1}^{\ell} (\bar{d}_i - 3)
= \sum_{i = 1}^{\ell} (d_i - 5)
= 4n - 6 - (4 \multipl_4 + 3 \multipl_3 + 2 \multipl_2) - 5\ell
= n - 6 - \multipl_4 + \multipl_2 - 2\ell
\end{equation}
%
By assumption~\eqref{assumption:omega-two} of the lemma, $\multipl_2 > 2$, hence also
$\multipl_3 + \multipl_2 \geq 3$, which can be written equivalently as
$n - \multipl_4 - 3 \geq \ell$. Equation~\eqref{eq:in claim}, hence $\sum_{i = 1}^{\ell} (\bar{d}_i - 3) \geq \multipl_2-\ell - 3 =q-\ell - 1$,
and the claim follows.
\end{proof}

Concluding that Lemma~\ref{lem:OP-4n-1} applies,
and noting that by construction $\sum {\hat d}=4{\hat n}-6$,
it follows that $\hat d$ is maximal outerplanaric.
Let $\hat G$ be a a maximal outerplanar realization of $\hat d$.
We add $q$ isolated vertices
to $V(\hat G)$
to generate an outerplanar realization $G'' = (V,E'')$
of $d''$.
Let $V = \{ v_1, v_2, \ldots, v_n \}$ such that $\deg_{G''}(v_i) = d''_i$, for $i = 1,2,\ldots,n$.
Hence, the set
$A := \{ v_{n - q + 1}, \ldots, v_{n} \}$
contains the isolated vertices of degree $0$ in $G''$.
For ease of notation, we relabel the vertices
$a_i := v_{n-q+i}$, for $i = 1,\ldots q$. Consequently,
\[
A ~=~ \{ a_{1}, \ldots, a_{q} \} = \{ v_{n - q + 1}, \ldots, v_{n} \}
~.
\]
Consider an outerplanar embedding of $\hat G$.
%
Let $\pi$ be the permutation by which the vertices of
$B := \{ v_1,v_2, \ldots,v_{q+1} \}$ appear
on the cyclic-order defined by the planar embedding of $G''$.
Without loss of generality, let $\pi(1) = 1$.
It follows that $G''$ has an embedding into a book page where the vertices
of $B$ are arranged along the spine according to $\pi$.
The vertices in $A$ can be placed freely along the spine without interfering
with the embedding of $G''$.

Note that unlike in case A, here we did not define the sequence $d''$
explicitly, so the sequence $d'$ too is not defined explicitly, but rather
is determined as $d' := d \ominus d''$.

To complete the proof, we
construct an outerplanar realization $G' = (V',E')$ of $d'$ such that $V' = V$
and $\deg_{G'}(v_i) = d'_i$, for $i = 1,2,\ldots,n$.
Only the vertices in $A$ and $B$ 
need to be connected in $G'$. Partition the vertices in $B$ into the sets
\begin{align*}
B_1 & = \{ v_{\ell+1}, \ldots, v_{q+1} \}
      & \text{ and } &&
B_2 & = \{ v_1,\ldots, v_{\ell} \}.
\end{align*}
Note that $d'_i = 1$ for $v_i \in B_1$, and $d'_i \geq 2$ for $v_i \in B_2$,
and that $d'_i = 2$ for every $v_i \in A$.
$G'$ and its book embedding are constructed simultaneously.
The vertices in $B$ are arranged according to $\pi$ along the book spine.
Hence, vertices of $B_1$ and $B_2$ may appear in any order and vertices in $A$
can be placed freely between them.
%
Since vertices in $A$ have degree $0$ in $G''$,
it follows that $E''\cap E'=\emptyset$,
which is necessary for $G'$ and $G''$ to form a valid book embedding
of a simple graph (rather than a multigraph).

It is essential that each edge in $E'$ is incident to one vertex in $A$ and
one vertex in $B$, i.e., $G'$ is a bipartite graph with parts $A$ and $B$.
Partition $A$ into $A_1\cup A_2$, where
$|A_1|=|B_1| = q - \ell + 1$ and $|A_2|=|B_2|-1 = \ell -1$.
Specifically, let
\begin{align*}
A_1 & =\{a_1,\ldots,a_{q-\ell+1}\}
      & \text{ and } &&
A_2 & =\{a_{q-\ell+2},\ldots,a_q\}.
\end{align*}
The relationship between the volume of $B_2$ and the cardinalities of
$A_2$ and $A_1$ is established in the following claim, which helps to show
that $G'$ is bipartite.

\begin{claim}
\label{claim:sup-eq-dem}
$\sum_{i = 1}^{\ell} d'_i = 2 |A_2| + |A_1|$.
\end{claim}

\begin{proof}
Notice that $\sum_{i = 1}^{\ell} d'_i$ consists of two parts.
One comes from the decrease from $d$ to $\bar{d}$ on the elements of $B_2$,
and is equal to $2\ell$;
the other part comes from the decreases of $d''$ on the elements of $B_2$
in Algorithm~\ref{alg:dprimeprime}, and is equal to $q-\ell-1$
by \eqref{eq:decrease-volum}.
Also recall that $|A_1| = q - \ell + 1$ and $|A_2| = \ell - 1$.
It follows that
$\sum_{i = 1}^{\ell} d'_i = 2 \ell + (q - \ell - 1) =
2( \ell - 1) + q - \ell + 1 = 2 |A_2| + |A_1|$.
\end{proof}

By~\Cref{claim:sup-eq-dem}, there is a bipartite graph $Bar(A,B)$
between $A$ and $B$, defined as follows.
\begin{enumerate}
\item Connect a matching between $B_1$ and $A_1$;
\item Connect the vertices of $B_2$ to the vertices of $A_1$ by constructing a
        union of disjoint stars, with $v_1$ connected to $d'(v_1)-1$ vertices in $A_1$,
        $v_i$ connected to $d'(v_i)-2$ vertices in $A_1$ for $2\leq i\leq \ell-1$,
        and $v_\ell$ connected to $d'(v_{\ell})-1$ vertices in $A_1$;
\item Construct a path $v_{1}-a_{q-\ell+2}-v_{2}-\ldots-a_{q}-v_{\ell}$
        zigzagging between $A_2$ and $B_2$.
\end{enumerate}
The \Cref{fig:bipartite-graph} illustrates the structure of $Bar(A,B)$.

One can easily verify that assuming no restrictions on the ordering
of the vertices, it is possible to find an outerplanar embedding of $Bar(A,B)$
in the plane. The only remaining obstacle is that the embedding must obey
the ordering $\pi$ dictated for the vertices of $B$ by the graph $G''$.
The remainder of the proof concerns finding an outerplanar embedding
obeying this ordering.

The construction of $G'$ consists of three steps corresponding to those of
the construction of $Bar(A,B)$.
First, construct a matching between $A_1$ and $B_1$.
Next, construct disjoint stars between $B_2$ and $A_1$.
Finally, construct the path between $A_2$ and $B_2$.

%

Our starting point is the spine defined for $G''$, consisting of
the vertices of $B$. This spine dictates the order of the $B$ vertices, which in turn
dictate the bipartite $A$-$B$ connections in the constructed graph $G'$.
It is convenient to separate the indices of $B_1$ and $B_2$.
Let $(h_1, \ldots, h_{q - \ell + 1})$ denote the increasing sequence of indices
$\{ i ~\mid~ v_{\pi(i)} \in B_1 \}$,
let $(k_1, \ldots, k_{\ell})$ denote the increasing sequence of indices
$\{ i ~ \mid ~ v_{\pi(i)} \in B_2 \}$.
The placement of the $A_1$ (respectively, $A_2$) vertices on the spine
is fixed in Step 1 (resp., Step 3) based on these positions of the
$B_1$ (resp., $B_2$) vertices.

In the first step, place one vertex of $A_1$ next to each vertex of $B_1$
in the embedding of $G'$ (and $G$) and connect them.
We place $a_i$
to the right of $v_{\pi(h_i)}$,
for $i = 1,\ldots, q-\ell+1$,
and add the edge $(a_{i},v_{\pi(h_i)})$
to $E'$.
%
See Figure~\ref{fig:book-emb1} for an illustration.

Note that after step 1, each of the vertices in $A_1$ is still one edge short,
while each of the vertices in $A_2 := A \setminus A_1$ is still missing two edges.
The goal of step 2 is to handle the remaining shortages of the vertices
of $A_1$ by connecting them to the vertices of $B_2$. Towards that,
%
%
we partition the vertices in $A_1$ into subsets {of vertices
that fall in the gaps between consecutive vertices of $B_2$. Formally,} define
$A_1^i = \{ a_{j} \in A_1 ~ : ~  k_i < j < k_{i+1} \}$,
for $i = 1, \ldots, \ell$, as the vertices of $A_1$ that appear between
$v_{\pi(k_i)}$ and $v_{\pi(k_{i+1})}$ along the spine.
%
Consider Figure~\ref{fig:book-emb1} where $k_1 = 1$, $k_2 = 4$, and $A_1^1 = \{ 2,3\}$.
%

For the second step, we make use of the following tool.

In the second step towards building the graph $G'$, we connect vertices
of $B_2$ to the vertices of $A_1$ by using the bipartite bracket balancing
algorithm \BBB.
This is done as follows.
Let $P=A_1$. For the other side, identify with every vertex $v_{\pi(k_i)} \in B_2$
a sequence of $s_i$ \emph{``ports''}, or connection points,
$\langle v_{\pi(k_i)}^1, \ldots, v_{\pi(k_i)}^{s_i}\rangle$, where
\[ s_i ~=~ d'(v_{\pi(k_i)}) -
\begin{cases}
1,  &  i=1, \\
2,  &  2\le i\le \ell-1, \\
1,  &  i=\ell,
\end{cases}
\]
and let $Q$ be the union of those sets, i.e., the set of all
resulting ports,
\[
Q=\langle v_{\pi(k_1)}^1,\ldots v_{\pi(k_1)}^{s_1}, v_{\pi(k_i)}^1, \ldots, v_{\pi(k_i)}^{s_i},\ldots, v_{\pi(k_{\ell})}^1,\ldots, v_{\pi(k_{\ell})}^{s_{\ell}}\rangle.
\]

The ordering of the merged set $R$ is dictated by the order of the vertices
of $A_1$ and $B_2$ on the spine.
Figure \ref{fig:connecting-B2-A1} illustrates the transformation to the
instance of \textsc{Bipartite Bracket Balancing} problem.

After applying Algorithm \BBB\ to this input, transform its output
into a collection of edges on our original graph as follows.
For each $1\leq i\leq \ell$,
merge the $s_i$ ``ports'' corresponding to $v_{\pi(k_i)}$ into the point
$v_{\pi(k_i)}$ and change all the edges that are incident to those $s_i$ ``ports''
to be incident to $v_{\pi(k_i)}$. This yields
an outerplanar graph $\bar{G}'$.
Figure \ref{fig:connecting-B2-A1-solution} illustrates the
resulting graph after applying the algorithm.

After the construction of $\bar{G}'$, the end vertices $v_{\pi(k_1)}$ and
$v_{\pi(k_{\ell})}$ of $B_2$ are still missing one edge each, and the vertices
$v_{\pi(k_{i'})}$ for $2\leq i'\leq k_{\ell-1}$ are still missing two edges each.
Our third step involves spreading the $A_2$ vertices between the $B_2$ vertices.
Each $A_2$ vertex is placed between two consecutive $B_2$ vertices
(according to $\pi$). Specifically, we place $a_{q-\ell+1+i}$ between
$v_{\pi(k_i)}$ and $v_{\pi(k_{i+1})}$ for $i\in [\ell-1]$, and construct the path
$$v_{\pi(k_1)}-a_{q-\ell+2}-v_{\pi(k_2)}-a_{q-\ell+3}-\ldots a_{q}-v_{\pi(k_{\ell})}.$$
There are two difficulties involved in performing this step.
The first concerns determining the exact position in which the vertex
$a_{q-\ell+1+i}$ should be placed between $v_{\pi(k_i)}$ and $v_{\pi(k_{i+1})}$.
The second is the need to eliminate possible crossings of edges
after adding the path above.

Let $N_{\bar{G}'}(v)=\{v'~\mid~ (v,v')\in E(\bar{G}')\}$ for any $v\in V$.
To handle the first difficulty, $a_{q-\ell+1+i}$ should be placed between
two vertices $v$ and $v'$ in $A_{1}^{i}\cup\{v_{\pi(k_i)},v_{\pi(k_{i+1})}\}$
such that the point $v$ to the left of $a_{q-\ell+1+i}$ satisfies that
any vertex in $N_{\bar{G}'}(v)$ is to the left of $v$, and
the point $v'$ to the right of $a_{q-\ell+1+i}$ satisfies that
any vertex in $N_{\bar{G}'}(v')$ is to the left of $v'$.
As shown in Figure \ref{fig:correction}, $a_{q-\ell+1+i}$ is placed between
$a_{i'+1}$ and $v_{\pi(k_{i+1})}$ since $a_{i'+1}$ has all its neighbors on its left
and $v_{\pi(k_{i+1})}$ has all its neighbors on its right.

To handle the second difficulty, we first present some notations. The set
$E_{le}^{i}$ (respectively, $E_{ri}^{i}$) contains edges connecting a vertex
$v_{\pi(j)}$, for $j < i$ (resp., $j > i+1$)
to a vertex in $A_1^{i}$.
Formally,
\begin{eqnarray*}
E_{le}^i &=& \{(v_{\pi(k_{j'})},a_j)\in E(\bar{G}')~\mid~ j'<i, ~a_{j}\in A_1^i\},
\\
E_{ri}^i &=& \{(v_{\pi(k_{j'})},a_j)\in E(\bar{G}')~\mid~ j'>i+1, ~a_{j}\in A_1^i\}.
\end{eqnarray*}
As shown in Figure \ref{fig:correction},
$E_{le}^{i}=\{(N'(a_{i'}),a_{i'}),(N'(a_{i'+1}),a_{i'+1})\}$~ and
~$E_{ri}^{i}=\emptyset$. Note that $N'(a_{i'})$ is equal to $v_{\pi(k_{j'})}$ with some $j'<i$. Similar result can apply for $N'(a_{i'+1})$.
%
When adding edges $(v_{\pi(k_i)},a_{q-\ell+i+1})$ and $(a_{q-\ell+i+1},v_{\pi(k_{i+1})})$,
they may cross the existing edges in $\bar{G}'$.
The following steps of Algorithm \ref{alg:correction} eliminate the problem
by replacing the set of crossing edges with noncrossing set of edges
while preserving the vertex degrees.

Edge crossing happens when $E_{le}^i\neq \emptyset$ in line
\ref{li:case1} or $E_{ri}^i\neq \emptyset$ in line \ref{li:case2}.
Notice that the right (respectively left) endpoints of edges in
$E_{le}^{i}$ (respectively $E_{ri}^{i}$) are in consecutive order by Algorithm \BBB.
Let the set of the right endpoints of edges in $E_{le}^{i}$ be
$\{a_{i'}, \ldots, a_{i'+\ell'}\}$ with $\ell'=|E_{le}^{i}|-1$.
Similarly, let the set of the left endpoints of edges in $E_{ri}^{i}$ be
$\{a_{i''}, \ldots, a_{i''+\ell''}\}$ with $\ell''=|E_{ri}^{i}|-1$.
Since $N_{\bar{G}'}(v)$ consists only one element for any
$v\in \{a_{i'}, \ldots, a_{i'+\ell'}, a_{i''}, \ldots, a_{i''+\ell''}\}$,
let $N'(v)$ refer to the only element in $N_{\bar{G}'}(v)$.
In the case where $E_{le}^i\neq \emptyset$, the edge $(v_{\pi(k_i)},a_{q-\ell+i+1})$
crosses all the edges in $E_{le}^i$. The correction process exchanges the right
endpoints of edges in $E_{le}^{i}\cup \{(v_{\pi(k_i)},a_{q-\ell+i+1})\}$ such that
there are no cross edges. In the other case, where $E_{ri}^i\neq \emptyset$,
$(a_{q-\ell+i+1}, v_{\pi(k_{i+1})})$ crosses all the edges in $E_{ri}^{i}$.
The correction process is similar to that of $E_{le}^i\neq \emptyset$.

Figure \ref{fig:correction} illustrates the correction process between
$v_{\pi(k_i)}$ and $v_{\pi(k_{i+1})}$ for the case where $E_{le}^i\neq \emptyset$.
Figure \ref{fig:final-graph} illustrates the corrected graph following the
entire correction process performed by applying Algorithm \ref{alg:correction},
based on Figure \ref{fig:connecting-B2-A1-solution}.

\begin{algorithm}
\caption{\small\sf Construction algorithm for $G'$ third step}
\label{alg:correction}
\SetKwInOut{Input}{input}
\SetKwInOut{Output}{output}
\Input{outerplanar graph $\bar{G'}$ obtained from the second step}
\Output{an outerplanar graph $G'$}
\ForAll{$i \in [\ell-1]$}{
	Put $a_{q-\ell+i+1}$ between two vertices $v$ and $v'$ in $A_{1}^{i}\cup\{v_{\pi(k_i)},v_{\pi(k_{i+1})}\}$
such that the point $v$ to the left of $a_{q-\ell+1+i}$ satisfies that
any vertex in $N_{\bar{G}'}(v)$ is to the left of $v$, and
the point $v'$ to the right of $a_{q-\ell+1+i}$ satisfies that
any vertex in $N_{\bar{G}'}(v')$ is to the left of $v'$.\label{li:position} \\
	Connect $(v_{\pi(k_i)},a_{q-\ell+i+1}), (a_{q-\ell+i+1},v_{\pi(k_{i+1})})$\label{li:path} \\
	\label{li:case1}
	\If {$E_{le}^i\neq \emptyset$}{
		Delete all the edges in $E_{le}^{i}$ and $(v_{\pi(k_i)},a_{q-\ell+i+1})$ \\
		Connect $(N'(a_{i'+\ell'}),a_{q-\ell+i+1}), (N'(a_{i'+\ell'-1}), a_{i'+\ell'}), \ldots, (N'(a_{i'}),a_{i'+1}), (v_{\pi(k_{i})},a_{i'})$ \\
        }
	\label{li:case2}
        \If {$E_{ri}^i\neq \emptyset$}{
		Delete all the edges in $E_{ri}^{i}$ and $(a_{q-\ell+i+1},v_{\pi(k_{i+1})})$ \\
		Connect $(N'(a_{i''}),a_{q-\ell+i+1}), (N'(a_{i''+1}),a_{i''}), \ldots, (N'(a_{i''+\ell''}), a_{i''+\ell''-1}), (v_{\pi(k_{i+1})},a_{i''+\ell''})$ \\
        }
}
\end{algorithm}

\begin{claim}
\label{obs:non-cross}
The output $G'$ of Algorithm~\ref{alg:correction} is outerplanar
and satisfies $\deg(G') = d'$.
\end{claim}

\begin{proof}
Let us first prove the outerplanarity of $G'$.
Consider a pair of vertices $v_{\pi(k_i)}, v_{\pi(k_{i+1})}$.
If $E_{le}^i=\emptyset$ and $E_{ri}^i=\emptyset$, then
for any edge $e\in E(\bar{G}')$ with one endpoint in $A_{1}^i$,
the other endpoint of $e$ is either $v_{\pi(k_i)}$ or $v_{\pi(k_{i+1})}$
by the definition of $E_{le}^i$ and $E_{ri}^i$.
Combining this with the choice of the position of $a_{q-\ell+i+1}$
in line \ref{li:position} of the algorithm,
$(v_{\pi(k_i)},a_{q-\ell+i+1})$ and $(v_{\pi(k_{i+1})},a_{q-\ell+i+1})$ do not cross
any edges with one endpoint in $A_{1}^{i}$.
Therefore, there are no crossing edges in this case.

If $E_{le}^i\neq \emptyset$, then $(v_{\pi(k_i)},a_{q-\ell+i+1})$ crosses
all the edges in $E_{le}^{i}$, which is
the edge set with one endpoint before $v_{\pi(k_i)}$ and one in $A_{1}^i$.
In this case, the correction process deletes the edges in
$E_{le}^{i}$ and $(v_{\pi(k_i)},a_{q-\ell+i+1})$
and connects
$(N'(a_{i'+\ell'}),a_{q-\ell+i+1})$, $(N'(a_{i'+\ell'-1}), a_{i'+\ell'})$, $\ldots, (N'(a_{i'}),a_{i'+1}), (v_{\pi(k_{i})},a_{i'})$,
which means that all the edges in $E_{le}^{i}$
change their right endpoint to the closest vertex on the right of it in $A$,
and $v_{\pi(k_i)}$ connects to the leftmost endpoint of edges in
$E_{le}^{i}$.
After this step, the newly added edges do not cross each other or the edges
in $E(\bar{G}')\setminus E_{le}^{i}$ since $\bar{G}'$ is outer planar.
The argument for $E_{ri}^i\neq \emptyset$ is similar to that of
$E_{le}^i\neq \emptyset$.

Next, we show that above correction processes are independent.
Notice that the correction process in Algorithm \ref{alg:correction} is
based on exchanging the endpoint of $e$ in $A_{1}^i$ to another vertex
in $A_{1}^i$ for some $e\in \bar{G}'$.
Therefore, the correction processes applied to different pairs of
$v_{\pi(k_i)}, v_{\pi(k_{i+1})}$ are independent.
Also, in the case that $E_{le}^i\neq \emptyset$,
Algorithm \ref{alg:correction} exchanges the right endpoints of edges
in $E_{le}^{i}$ and $a_{q-\ell+i+1}$.
Similarly, in the case that $E_{ri}^i\neq \emptyset$,
Algorithm \ref{alg:correction} exchanges the left endpoints of edges
in $E_{ri}^{i}$ and $a_{q-\ell+i+1}$.
Since $E_{le}^{i}$ and $E_{ri}^{i}$ are disjoint,
the correction processes for $E_{le}^i\neq \emptyset$ and
$E_{ri}^i\neq \emptyset$ are independent.
In conclusion, after performing all of the above corrections
for each pair $v_{\pi(k_i)}, v_{\pi(k_{i+1})}$, there are no crossing edges
left in $G'$.

Finally, we show that there are no parallel edges in the resulting $G'$.
In the process of exchanging the endpoints, in the case that
$E_{le}^i\neq \emptyset$, the added edges are $(N'(a_{i'+\ell'}),a_{q-\ell+i+1})$, ~ $(N'(a_{i'+\ell'-1}), a_{i'+\ell'})$, ~$\ldots$, $(N'(a_{i'}),a_{i'+1})$, ~ $(v_{\pi(k_{i})},a_{i'})$,
which did not exist in $\bar{G}'$.
A similar argument applies for $E_{ri}^i\neq \emptyset$.
Therefore, there are no parallel edges.
It follows that $G'$ is a simple outerplanar graph.

To prove that $\deg(G') = d'$, it suffices to show that $G'$ has the same
degree sequence as the bipartite graph $Bar(A,B)$, since $\deg(Bar(A,B))=d'$.
In the first and second step, the construction of $G'$ is same as that of
$\bar{G}'$. Comparing $\bar{G}'$ with $Bar(A,B)$, the vertices $v_{\pi(k_1)}$ and
$v_{\pi(k_{\ell})}$ are each one edge short, while the vertices $v_{\pi(k_{i'})}$
for $2\leq i'\leq \ell-1$ are each two edges short.
Every vertex in $A_2$ has degree 0.
In the third step, the correction process in Algorithm \ref{alg:correction}
exchanges the endpoints of the cross edges, which keeps the degree $2$
of the vertices in $A_1$. Also, the degrees of $v_{\pi(k_1)}$ and $v_{\pi(k_{\ell})}$
each increase by 1, and the degrees of $v_{\pi(k_i)}$ for $2\leq i\leq \ell-1$
each increase by 2. For $1\leq i\leq \ell$, the vertex $a_{q-\ell+i+1}$
inserted between $v_{\pi(k_{i})}$ and $v_{\pi(k_{i+1})}$ has degree 2 in $G'$.
Therefore, $\deg(G')=\deg(Bar(A,B))$.
\end{proof}

In summary, $G'$ has an embedding into the page of a book where the
vertices are ordered along book spine according to $\pi$.
It follows that $G$ and $G'$ form a 2PBE realization of $d$.
Note that $G$ is outerplanar and $G'$ is bipartite outerplanar.
This completes the proof of Case B, and of Lemma~\ref{lem:BE}.
\end{proof}

Lemma~\ref{lem:BE} considers maximal-volume sequences.
The following lemma mainly focuses on non-maximal-volume sequences.
Recall that $\edgeD(d) = (4n - 6 - \sum d) / 2$ and
$\omegaS(d)=  3\multipl_1+2 \multipl_2 + \multipl_3 - n$.

\begin{lemma}
\label{lem:BE-3}
Let $d\in\cD_{\ge 5}$ and $d_1\le n-1$.
If
\begin{inparaenum}[(1)]
\item $2\edgeD \geq \omegaS$,
\item $\multipl_2 > 2$
\item $d_2 > 3$, and
\item $d_n = 2$,
\end{inparaenum}
Then $d$ has
an OP+bi realization.
\end{lemma}

\begin{proof}
Given $d$ as in the lemma, apply to it the following Algorithm~\ref{alg:trans};
let $\bar{d}$ be its output.

\begin{algorithm}[hbt]
\caption{\small generating $\bar{d}$ from $d$ }
\label{alg:trans}
Set $\bar{d} \gets d$ \\
\While {$\sum \bar{d} < 4\bar{n} - 6$ AND $\bar{\multipl}_2 > 3$ AND $\bar{d}_1 > 4$}{
	$\bar{d}_j \gets \bar{d}_j - 2$, where $j = \text{argmax}_{i \in [\bar{n}]} ~ \bar{d}_i$ \label{line:red} \\
	Remove two degrees $2$ from $\bar{d}$ \\
}
\end{algorithm}

To prove the lemma, we first show that $\bar{d}$ has a realization $\bar{G}$ that is outerplanar or a book embedding with two pages.
Then, we modify $\bar{G}$ to generate a realization $G$ of $d$ with the analogous properties.

Assume Algorithm~\ref{alg:trans} executes $k$ iterations of its while-loop.
Let $\bar{d}^i$ be the intermediate sequence after iteration $i \in [0,k]$, i.e., $\bar{d}^0 = d$ and $\bar{d}^k = \bar{d}$.
%
Observe that
\begin{enumerate}[(i)]
	\item $\bar{n}^i = \bar{n}^{i+1} + 2$,
	\item $\bar{\multipl}^i_2 = \bar{\multipl}^{i+1}_2 + 2$, and
	\item $\sum \bar{d}^i = \sum \bar{d}^{i+1} + 6$.
\end{enumerate}
Next, we argue that $\edgeDbar^i = \edgeDbar^{i+1} + 1$.
Verify that,
\[
2\edgeDbar^i = 4\bar{n}^i-6 - \sum \bar{d}^i = 4\bar{n}^{i+1} - 6 - \sum \bar{d}^{i+1} + 2 = 2(\edgeDbar^{i+1} + 1).
\]
We proceed showing that $\omegaSbar^i = \omegaSbar^{i+1} + 2$.
Verify that,
\[
\omegaSbar^i = 2 \bar{\multipl}^i_2 + \bar{\multipl}^i_3 - \bar{n}^i =
2 \bar{\multipl}^{i+1}_2 + 4 + \bar{\multipl}^i_3 - \bar{n}^{i+1} - 2 =
\omegaSbar^{i+1} + 2.
\]
It follows that $2 \edgeDbar \geq \omegaSbar$. 
Moreover, note that $\bar{d}_{\bar{n}-1} = \bar{d}_{\bar{n}} = 2$ holds.

There are three conditions under which Algorithm~\ref{alg:trans} stops.
Each of them allows us to use a previous result and generate a realization of $\bar{d}$.
We consider the three cases separably.

\begin{description}
\item[Case (a):]
  $\bar{d}_1 \leq 4$. $\bar{d}_1\le \bar{n}-1$ by Algorithm~\ref{alg:trans}. Then Lemma~\ref{lem:2pbe-d1leq4},
provides a 2PBE realization
$(\bar{G},\bar{G}')$ of $\bar{d}$.


\item[Case (b):]
$\bar{d}_1 > 4$ and $\sum \bar{d} = 4\bar{n} - 6$.
Since $\edgeDbar = 0$ and $2\edgeDbar \geq \omegaSbar$ hold, we have that $2 \bar{\multipl}_2 + \bar{\multipl}_3 \leq \bar{n}$.
It follows that $\bar{d}$ satisfies the preconditions of Lemma~\ref{lem:BE}
yielding a 2PBE realization
$(\bar{G},\bar{G}')$.

\item[Case (c):]
$\bar{d}_1 > 4$, $\sum \bar{d} < 4\bar{n} - 6$, and $\bar{\multipl}_2 \leq 3$.
If $\bar{\multipl_2} = 2$, then $\bar{d}$ has an outerplanar realization $\bar{G}$ due to Lemma~\ref{lem:OP-4n-1}.
Otherwise, $\bar{\multipl}_2 = 3$ holds, and Corollary~\ref{lem:OP-4n-2} provides an outerplanar realization $\bar{G}$ of $\bar{d}$.
\end{description}

In either of the cases, $\bar{G}$ is a 2PBE graph
that (partially) realizes $\bar{d}$.
During the execution of Algorithm~\ref{alg:trans}, the order of $\bar{d}$ is kept the same and it is not necessarily in nonincreasing order.
For $i \in [k]$, let $\ell_i$ be the index of the degree that was reduced by $2$ in iteration $i$ by Algorithm~\ref{alg:trans} (line\eqref{line:red}).
To finish the proof, we modify $\bar{G}$ by running Algorithm~\ref{alg:reverse}.
Intuitively, it reverses the reduction, which was applied to sequence $d$ by Algorithm~\ref{alg:trans}, on the graph $\bar{G}$.

\begin{algorithm}[h]
\caption{\small reversing $\bar{d}$ to $d$}
\label{alg:reverse}
Set $G \leftarrow \bar{G}$ \\
\For {$i \in [k]$}{
	Let $v \in V(G)$ such that $\deg_G(v) = \bar{d}_{\ell_{k-i+1}}$ \\
	$V(G) \gets V(G) \cup \{u,w\}$ \\
	Add edges $(v,u), (v,w)$, and $(u,w)$ to $E(G)$ \\
	$\bar{d}_{\ell_{k-i+1}} = \bar{d}_{\ell_{k-i+1}} + 2$ \\
}
\end{algorithm}

Algorithm~\ref{alg:reverse} starts constructing graph $G$ by copying $\bar{G}$.
Note that $\bar{G}$ is guaranteed to have a vertex $v$ with degree $\bar{d}_{\ell_{k}}$.
In the last iteration of Algorithm~\ref{alg:trans}, $\bar{d}_{\ell_{k}}$ was reduced by $2$.
The degree of $v$, which is $\bar{d}_{\ell_{k}}$, is increased by $2$
in the first iteration of Algorithm~\ref{alg:reverse}.
Additionally, two more vertices are added to $G$ which are connected to each other and to $v$. Figure~\ref{fig:book-emb3} shows how the new vertices can be added to an existing outerplanar embedding of $G$.

It follows that the graph $G$ returned by Algorithm~\ref{alg:reverse} is a (partial) realization of $d$.
If $\bar{G}$ is an outerplanar realization of $\bar{G}$, then $G$ is an outerplanar realization of $d$.
Otherwise, $(\bar{G},\bar{G}')$ is a 2PBE realization of $\bar{d}$.
Observe that the vertices added by Algorithm~\ref{alg:reverse} can be added to $\bar{G}'$ as isolated vertices without invalidating the book embedding.
Hence, $(G,\bar{G}')$ is a 2PBE realization of $d$,
where $G$ is outerplanar and $\bar{G}'$ is bipartite outerplanar.
The lemma follows.
\end{proof}


\begin{lemma}
\label{small d2}
Every graphic $d\in \cD_{\ge 5}$ with $d_2\leq 3$ has an OP+bi realization.
\end{lemma}

\begin{proof}
First, in $\multipl_1$ steps, connect a new vertex of degree 1 with the vertex
of currently largest degree
obtaining the residual degree sequence $d'$.
Since $\sum d\geq 2n$ and $\sum d\leq 4n-2\multipl_1-6$,
\begin{align}\label{property-d'}
\sum d' & \leq 4(n-\multipl_1)-6, & d'_{n-\multipl_1} & \geq 2.
\end{align}
If $d'_1\leq 4$, then $d'$ has a 2PBE realization graph $G'$
by Theorem \ref{thm:d1leq4} and Inequality \eqref{property-d'}. Add $d_i-d'_i$ leaves to the vertex in $G'$ with degree $d'_i$ for $1\leq i\leq n-\multipl_1$ and we get a graph $G$. Notice that $G$ realizes $d$ and has page number 2.

Assume $u_1,\ldots,u_{n-\multipl_1}$ are the vertices with decreasing degree in $d'$ and $n'=n-\multipl_1$. If $d'_1\geq 5$, then $d'_1>d_2$. All the $\multipl_1$ leaves are connected to the vertex with $d'_1$ since the leaf is connected with the vertex of currently largest degree. Therefore, $d_1-d'_1=\multipl_1$. Then $d'_1\le n-1-\multipl_1$ since $d_1\le n-1$. Next connect $u_1$ with $u_2,\ldots u_{d'_1+1}$ and obtain the residual sequence $d''$. Notice that $d''$ has the form of $(2^{n'-1})$ or $(3^{\multipl'_3},2^{\multipl'_2})$ or $(2^{\multipl'_2},1^{\multipl'_1})$ with $\multipl'_3,\multipl'_2,\multipl'_1>0$.
\begin{description}
\item[Case 1:] $d''=(2^{n'-1})$. Since $d'_1\ge 5$ and $d_2\le 3$, there are $d'_1$ vertices changing their degree from 3 in $d'$ to 2 in $d''$. Then $n'-1\ge 5$ and $d''$ can be realized by a cycle.
\item[Case 2:] $d''=(3^{\multipl'_3},2^{\multipl'_2})$. Since $\multipl'_3>0$ and $d_2\le 3$, there are $d'_1$ vertices changing their degree from 3 in $d'$ to 2 in $d''$. Then $\multipl'_2\geq 5$. Combining it with that $\multipl'_3$ is even, $d''$ can be realized by an outerplanar graph .
\item[Case 3:] For $d''=(2^{\multipl'_2},1^{\multipl'_1})$, it can be realized by the union of a cycle and a matching or a forest. Therefore, $d''$ can be realized by an outerplanar graph $G''$.
\end{description}

From above three cases, $G"$ is an outer planar graph. Add an isolated vertex $u_1$ to $G''$ and obtain an outerplanar graph
$G_1'=G''\cup \{u_1\}=(V,E')$ with $V=V''\cup \{u_1\}$.
Let $G'_2= (V,E'')$ with $E''=\{(u_1,u_2),\ldots,(u_1,u_{d'_1+1})\}$.
The position of $u_1$ in the embedding of $G'_1$ and $G'_2$ can be fixed
arbitrarily since $u_1$ is isolated in $G'_1$, and $G'_2$ has edges
that do not exist in $G'$. $G'_1\cup G'_2$ realizes $d'$ with $2$ page embedding. Add $d_i-d'_i$ leaves to the vertex in $G'_2$ with degree $d'_i$ for $1\leq i\leq n-\multipl_1$ and an outer planar graph $G_2$ is obtained. Notice that $G'_1\cup G_2$ realizes $d$ and has 2 page book embedding,
where $G'_1$ is outerplanar and $G_2$ is bipartite outerplanar.
\end{proof}

\begin{corollary}
\label{cor:dn equal to 2}
Let $d\in\cD_{\ge 5}$ and $d_1\le n-1$.
If $2\edgeD \geq \omegaS$ and $d_n =d_{n-1}= 2$,
then $d$ has
an OP+bi realization.
\end{corollary}

\begin{proof}
If $\multipl_2 =2$, then by \Cref{lem:OP-4n-1}, $d$ is outerplanaric.
If $d_2\leq 3$, then $d$ has a 2PBE realization
by \Cref{small d2}. Combining above observations with \Cref{lem:BE-3}, this corollary is proved.
\end{proof}

\begin{lemma}
\label{outerplanar with degree 1}
Let $d\in\cD_{\ge 5}$ and $d_1\le n-1$.
If $2\edgeD \geq \omegaS$ and $d$ satisfies
$d_{n-1} \leq 2$ and $d_{n-2} \leq 3$,
then $d$ has
an OP+bi realization.
\end{lemma}

\begin{proof}
By the assumption on $d_{n-1}$, $d_n$ is either 1 or 2.
If $d_n=2$, then $d_{n-1}=2$ by
the assumption on $d_{n-1}$, implying that $d$ has a 2PBE realization
by Corollary \ref{cor:dn equal to 2}.
So hereafter we focus on the case where $d_n=1$.

We apply to $d$ a sequence of transformations described in Algorithm
\ref{removing degree 1}, which gradually get rid of the 1 degrees.

\begin{algorithm}[ht]
\caption{\small\sf Removing 1 degrees}
\label{removing degree 1}
Set $\bar{d} \leftarrow d$ \\
\While {$\bar{\multipl}_1 >0$ AND $\bar{d}_1 > 4$}{
	$\bar{d}_1 \gets \bar{d}_1 - 1$ \\
	Remove one degree $1$ from $\bar{d}$ \\
	Re-sort $\bar{d}$\\
}
\end{algorithm}

Assume Algorithm ~\ref{removing degree 1} executes $k$ iterations
of its while-loop before halting.
Let $\bar{d}^i$ be the intermediate sequence after iteration
$i \in [0,k]$, i.e., $\bar{d}^0 = d$ and $\bar{d}^k$ is
the sequence $\bar{d}$ reached upon termination of the algorithm.
Observe that
\begin{enumerate}[(i)]
\item $\bar{n}^i = \bar{n}^{i+1} + 1$,
\item $\bar{\multipl}^i_1 = \bar{\multipl}^{i+1}_1 + 1$,~~
$\bar{\multipl}^i_2= \bar{\multipl}^{i+1}_2$,~~
$\bar{\multipl}^i_3 = \bar{\multipl}^{i+1}_3$,~~
and
\item $\sum \bar{d}^i = \sum \bar{d}^{i+1} + 2$.
\end{enumerate}
These observations imply also that,
\begin{eqnarray}
2\edgeDbar^i &=& 4\bar{n}^i-6 - \sum \bar{d}^i
~=~ 4(\bar{n}^{i+1}+1) - 6 - (\sum \bar{d}^{i+1} + 2)
\nonumber
\\
&=& 4\bar{n}^{i+1} - 6 - \sum \bar{d}^{i+1} + 2
~=~ 2\edgeDbar^{i+1} + 2.
\label{eq: DeltaE change}
\end{eqnarray}
By the above observations (i) and (ii),
\begin{align}
\omegaSbar^i & = 3\bar{\multipl}^i_1+2 \bar{\multipl}^i_2 + \bar{\multipl}^i_3
- \bar{n}^i
= 3(\bar{\multipl}^{i+1}_1+1)+2 \bar{\multipl}^{i+1}_2 + \bar{\multipl}^{i+1}_3
- \bar{n}^{i+1} - 1
= \omegaSbar^{i+1} + 2
\label{eq: DeltaW change}
\end{align}
We next claim that throughout the execution of Algorithm
\ref{removing degree 1},
\begin{align}
2 \edgeDbar^i \geq \omegaSbar^i,
~~~~~~\mbox{for}~ i=0,\ldots,k,
\label{one property of algorithm}
\\
\sum \bar{d}^i\leq 4\bar{n}^i-6-2\bar{\multipl}^i_1,
~~~~~~\mbox{for}~ i=0,\ldots,k.
\label{the other property of algorithm}
\end{align}
Both claims can be shown by induction on $i$, where the basis
follows from the assumptions of the lemma. For the first inequality,
the inductive progress step follows by Eq.
\eqref{eq: DeltaE change} and \eqref{eq: DeltaW change}.
For the second inequality, the progress step follows from
the inductive hypothesis since
\[\sum d^{i+1} = \sum d^{i}+2\leq 4n^i-6-2\bar{\multipl}^i_1+2=4n^{i+1}-6-2\bar{\multipl}^{i+1}_1.\]

Algorithm \ref{removing degree 1} stops when either $\bar{\multipl}_1=0$
or $\bar{d}_1 \leq 3$. Let the final sequence be $d'$ and $n$ is updated
to $n'$.
By Eq. \eqref{one property of algorithm} and
\eqref{the other property of algorithm}, we have
\begin{eqnarray}
2 \edgeD' &\geq& \omegaS'~,
\label{eq: property 1}
\\
\sum d' &\leq& 4n'-6-2\multipl'_1~.
\label{eq: property 2}
\end{eqnarray}

Next, we claim that the resulting sequence $d'$ has a 2PBE realization.
To prove this, consider the above two cases separately.

\begin{description}
\item[Case a:]
$\multipl'_1=0$. By Eq.~\eqref{eq: property 2}, we have
  $d'_{n'-1}\leq 3$.
%
Consider two possible subcases,
$(d'_{n'-1}, d'_{n'})=(2,2)$ and $(d'_{n'-1}, d'_{n'}) \in \{(2,3),(3,3)\}$.

\item[SubCase (a1):] $d'_{n'}=d'_{n'-1}=2$. Then combining
Eq.~\eqref{eq: property 1} and \eqref{eq: property 2}, the conditions
in Corollary \ref{cor:dn equal to 2} are satisfied and $d'$ has a 2PBE realization $G'$.

\item[SubCase (a2):]
$d'_{n'}\in\{2,3\}$ and $d'_{n'-1}=3$.
Then, $d'$ has a 2PBE realization $G'$ by \Cref{lem:large dn}.
%

\item[Case (b):]
$d'_1\leq 4$ and $\bar{\multipl}_1>0$. Notice that $d'_1\le n'-1$ by Algorithm~\ref{removing degree 1}. In this case, $d'$ 2PBE realization $G'$
by \Cref{lem:2pbe-d1leq4 with 1}
and Eq.~\eqref{eq: property 2}.
\end{description}

It follows that in all cases, $d'$ has a 2PBE realization.
We now transform this realization into a 2PBE realization $G$ for the original sequence $d$ by creating $d_i-d'_i$
new vertices of degree 1 and connecting them to vertex $i$ in $G'$,
for every $i$. (Clearly, these vertices can be embedded on the existing
pages of $G'$ while preserving their outerplanarity.)
By the above construction, $G$ consists of two graphs, where one is outerplanar
and the other is bipartite outerplanar.
\end{proof}

Combining Lemmas~\ref{lem:op_nec_deg23},
\ref{lem: between edge and multiplicity}
and \ref{outerplanar with degree 1},
we conclude the following.

\begin{corollary}
\label{cor: D ge 5 2pbe}
There is a polynomial time algorithm that, given an outerplanaric sequence
$d\in\cD_{\ge 5}$ with $d_2 > 3$, returns
an OP+bi realization for $d$.
\end{corollary}

\subsection{Putting the pieces together}

Combining Thm. \ref{thm:d1leq4}, \Cref{small d2} and Cor. \ref{cor: D ge 5 2pbe},
we conclude with the following.

%

\begin{theorem}
There is a polynomial time algorithm that given a sequence $d\in\cD$ either
returns ``$d$ is not outerplanaric'' or returns a 2PBE realization for $d$.
\end{theorem}



\begin{corollary}
There is a poly-time algorithm
that constructs a 1PBE or 2PBE realization
for any sequence in a superclass\footnote{The class $\mathcal{S}$ strictly
contains the class $\mathcal{OP}$. For example, the sequence
$(4^2,2^4)$ is not outerplanaric yet it satisfies the conditions of
Lemma \ref{lem:2pbe-d1leq4},
hence our algorithm succeeds in constructing a 2PBE realization for it.}
$\mathcal{S}$ of
the class $\mathcal{OP}$ of all outerplanaric sequences.
\end{corollary}


\clearpage
\bibliographystyle{abbrv}
\bibliography{realizations}

\end{document}